\documentclass[11pt]{article}  
\usepackage{times}
\usepackage{amsmath, amsthm, amssymb}
\usepackage{xspace,stmaryrd,calc}

\usepackage{microtype}

\newcommand{\np}{{\em NP}\xspace}
\newcommand{\nphard}{\np-hard\xspace}

\newtheorem{thm}{Theorem}[section]

\newtheorem{claim}[thm]{Claim}

\newtheorem{theorem}[thm]{Theorem}
\newtheorem{lemma}[thm]{Lemma}
\newtheorem{corollary}[thm]{Corollary}
{\theoremstyle{remark} \newtheorem{remark}[thm]{Remark}}
{\theoremstyle{definition} \newtheorem{definition}[thm]{Definition} 
\newtheorem{algorithm}{Algorithm}}

\newenvironment{proofof}[1]{\begin{proof}[Proof of #1]}{\end{proof}}

\newenvironment{labellist}[1][A]
{\begin{list}{{#1}\arabic{enumi}.}{\usecounter{enumi}\addtolength{\leftmargin}{-1ex}
      \addtolength{\labelwidth}{\widthof{{#1}5}}}} 
{\end{list}}

\DeclareMathOperator{\E}{E}
\DeclareMathOperator{\poly}{poly}

\DeclareMathOperator{\argmin}{argmin}
\DeclareMathOperator{\conv}{conv}

\newcommand{\R}{\ensuremath{\mathbb R}}

\newcommand{\C}{\ensuremath{\mathcal{C}}}

\newcommand{\T}{\ensuremath{\mathcal T}}

\newcommand{\Pc}{\ensuremath{\mathcal P}}
\newcommand{\Qc}{\ensuremath{\mathcal Q}}
\newcommand{\W}{\ensuremath{\mathcal W}}
\newcommand{\OPT}{\ensuremath{\mathit{OPT}}}

\newcommand{\charge}{\ensuremath{\mathsf{charge}}}

\newcommand{\sm}{\ensuremath{\setminus}}
\newcommand{\es}{\ensuremath{\emptyset}}
\newcommand{\sse}{\subseteq}
\newcommand{\frall}{\ensuremath{\text{ for all }}}

\newcommand{\mlp}{\ensuremath{\mathsf{MLP}}\xspace}
\newcommand{\kmlp}{\ensuremath{k\text{-}}\mlp}
\newcommand{\cg}{\ensuremath{\mathit{CG}}}
\newcommand{\bnslb}{\ensuremath{\mathsf{BNSLB}}}
\newcommand{\lb}{\ensuremath{\mathsf{LB}}}
\newcommand{\optbns}[1]{\ensuremath{\mathsf{BNS}({#1})}}
\newcommand{\into}{\ensuremath{\mathrm{in}}}
\newcommand{\out}{\ensuremath{\mathrm{out}}}
\newcommand{\lmst}{\ensuremath{\ell\text{-}\mathsf{MST}}\xspace}
\newcommand{\pcst}{\ensuremath{\mathsf{PCST}}\xspace}

\newcommand\Time{\ensuremath{\mathsf{T}}}
\newcommand{\lat}{\ensuremath{\mathsf{lat}}}

\newcommand{\Pcol}{\ensuremath{\Pc}}
\newcommand{\Tcol}{\ensuremath{\T}}

\newcommand{\vP}{\ensuremath{\vec{P}}}
\newcommand{\vQ}{\ensuremath{\vec{Q}}}

\newcommand{\vpcol}{\ensuremath{\Pc}}
\newcommand{\vtcol}{\ensuremath{\T}}

\newcommand{\e}{\ensuremath{\epsilon}}
\newcommand{\ve}{\ensuremath{\varepsilon}}
\newcommand{\gm}{\ensuremath{\gamma}}
\newcommand{\Gm}{\ensuremath{\Gamma}}
\newcommand{\ld}{\ensuremath{\lambda}}
\newcommand{\Ld}{\ensuremath{\Lambda}}
\newcommand{\kp}{\ensuremath{\kappa}}
\newcommand{\al}{\ensuremath{\alpha}}
\newcommand{\tht}{\ensuremath{\theta}}

\newcommand{\dt}{\ensuremath{\delta}}
\newcommand{\Dt}{\ensuremath{\Delta}}
\newcommand{\sg}{\ensuremath{\sigma}}

\newcommand{\TS}{\ensuremath{\mathsf{TS}}}
\newcommand{\bb}[1]{\ensuremath{\llbracket {#1} \rrbracket}}
\newcommand{\iopt}{\ensuremath{O^*}}
\newcommand{\bT}{\ensuremath{T}}
\newcommand{\assign}{\ensuremath{\leftarrow}}
\newcommand{\wt}{\ensuremath{\mathsf{wt}}}

\newcommand{\lppathts}{\ensuremath{(\text{LP}^\TS_\Pcol)}}
\newcommand{\lptreets}{\ensuremath{(\text{LP}^\TS_\Tcol)}}

\newcommand{\ceil}[1]{\ensuremath{\left\lceil#1\right\rceil}}
\newcommand{\floor}[1]{\ensuremath{\left\lfloor#1\right\rfloor}}

\title{Linear-Programming based Approximation Algorithms for Multi-Vehicle Minimum Latency Problems}

\author{
    Ian Post\thanks{{\tt ian@ianpost.org, cswamy@uwaterloo.ca}.  
    Dept. of Combinatorics and
    Optimization, Univ. Waterloo, Waterloo, ON N2L 3G1.  Supported in part by NSERC
    grant 327620-09.  The second author is also supported by an NSERC Discovery
    Accelerator Supplement Award and an Ontario Early Researcher Award.}  
\and
\addtocounter{footnote}{-1} 
    Chaitanya Swamy\footnotemark 
}

\date{}

\begin{document}

\maketitle

\begin{abstract}
We consider various {\em multi-vehicle versions of the minimum latency problem}. There is a
fleet of $k$ vehicles located at one or more depot nodes, and we seek a collection of routes
for these vehicles that visit all nodes so as to minimize the total latency incurred,
which is the sum of the client waiting times. We obtain an $8.497$-approximation for
the version where vehicles may be located at multiple depots and a $7.183$-approximation
for the version where all vehicles are located at the same depot, both of which are the
first improvements on this problem in a decade. Perhaps more significantly, our algorithms 
exploit various LP-relaxations for minimum-latency problems. 
We show how to effectively leverage two classes of LPs---{\em configuration LPs} and 
{\em bidirected LP-relaxations}---that are often believed to be quite powerful but have
only sporadically been effectively leveraged for network-design and vehicle-routing
problems.   
This gives the first concrete evidence of the effectiveness of LP-relaxations for this
class of problems.

The $8.497$-approximation the multiple-depot version is obtained by rounding a near-optimal
solution to an underlying configuration LP for the problem. 
The $7.183$-approximation can be obtained both via rounding a bidirected LP for the
single-depot problem or via more combinatorial means. The latter approach uses a bidirected LP
to obtain the following key result that is of independent interest: 
for any $k$, we can efficiently compute a
rooted tree that is at least as good, with respect to the prize-collecting objective
(i.e., edge cost + number of uncovered nodes) as the best collection of $k$ rooted
paths. This substantially generalizes a result of Chaudhuri et al.~\cite{ChaudhuriGRT03}
for $k=1$, yet our proof is significantly simpler. 
Our algorithms are versatile and extend easily to handle various extensions involving: (i)
weighted sum of latencies, (ii) constraints specifying which depots may serve which nodes,
(iii) node service times.   

Finally, we propose a configuration LP that sheds further light on the power of
LP-relaxations for minimum-latency problems. 
We prove that the integrality gap of this LP is at most $3.592$, even for the multi-depot
problem, both via an efficient rounding procedure, and by showing that it is at least as
powerful as a stroll-based lower bound that is oft-used for minimum-latency problems; the
latter result implies an integrality gap of at most $3.03$ when $k=1$.  
Although, we do not know how to solve this LP in general, it can be solved
(near-optimally) when $k=1$, and this yields an LP-relative $3.592$-approximation for the
single-vehicle problem, matching (essentially) the current-best approximation ratio for
this problem.  
\end{abstract}

\section{Introduction} \label{intro}
Vehicle-routing problems constitute a broad class of combinatorial-optimization
problems that find a wide range of applications and have been widely studied in the
Operations Research and Computer Science communities (see, e.g.,~\cite{TothV02}). 
A fundamental vehicle-routing problem is the {\em minimum latency problem} (\mlp),
variously known as the {\em traveling repairman problem} or the 
{\em delivery man problem}~\cite{AfratiCPPP86,Minieka89,FischettiLM93,BC+94}, 
wherein, taking a client-oriented perspective, we seek a route starting at a given root
node that visits all client nodes and minimizes the total client waiting time.
We consider various multi-vehicle versions of the minimum latency problem (\mlp). In these 
problems, there is a fleet of $k$ vehicles located at one or more depot nodes, and we seek
a collection of routes for these vehicles that together visit all the client nodes so as
to minimize the total latency incurred, which is the sum of the client waiting times. 

Formally, we consider the {\em multi-depot $k$-vehicle minimum latency problem}
(multi-depot \kmlp), which is defined as follows. We are given a complete undirected graph
$G=(V,E)$ on $n$ nodes, with metric edge costs $\{c_e\}$, and $k$ not necessarily distinct
root/depot nodes $r_1,\ldots,r_k\in V$. 
A feasible solution consists of $k$ paths $P_1,\ldots,P_k$, where each path $P_i$
starts at root $r_i$, 
such that the $P_i$s cover all the nodes. 
The waiting time or {\em latency} of a node $v$ that is visited by path $P_i$, is the 
distance from $r_i$ to $v$ along $P_i$, and is denoted by $c_{P_i}(v)$. 
The goal is to minimize the {\em total latency} 
$\sum_{i=1}^k\sum_{v\in P_i:v\neq r_i}c_{P_i}(v)$ incurred.%
\footnote{Multi-depot \kmlp is often stated in terms of finding $k$ tours
starting at $r_1,\ldots,r_k$; 
since the last edge 
on a tour does not contribute to the latency of any node, the two formulations are
equivalent. We find the path-formulation to be more convenient.} 
(Due to metric costs, one may assume that any two $P_i$s are node disjoint, unless they
share a common root, which is then the only shared node.)
We refer to the special case where all depots are identical, i.e., $r_1=r_2=\ldots=r_k$, as
{\em single-depot \kmlp}, which we abbreviate simply to \kmlp.

In addition to self-evident applications in logistics, the problem of finding an  
optimal routing as modeled by multi-depot \kmlp 
can also be motivated from the perspective of searching a graph (e.g., the web graph) for a
hidden treasure~\cite{BC+94,KoutsoupiasPY96,AusielloLM00}; if the treasure is placed at a
random node of the graph then multi-depot \kmlp captures the problem of minimizing the
expected search time using $k$ search agents (e.g., web crawlers).
Even 
1-\mlp is known to be {\em MAXSNP}-hard for
general metrics~\cite{BC+94,PapadimitriouY93} and \nphard for tree metrics~\cite{Sitters02},
so we focus on approximation algorithms. 

\vspace{-1ex}
\paragraph{Our results and techniques.}
We obtain approximation guarantees of $8.497$ for multi-depot \kmlp
(Theorem~\ref{multithm}) and 
$7.183$ for (single-depot) \kmlp (Theorem~\ref{kmlpthm}),
which are the 
 first improvements on the respective problems in over a decade. 
The previous best approximation ratios for the multi- and single- depot problems were
$12$ (by combining~\cite{ChekuriK04,ChaudhuriGRT03}) and $8.497$ (by
combining~\cite{FakcharoenpholHR03,ChaudhuriGRT03}; see also~\cite{FakcharoenpholHR07}) 
respectively.    

Perhaps more significantly, as we elaborate below, 
our algorithms exploit various linear-programming (LP)
relaxations, including various {\em configuration-style LP relaxations} as well as a  
{\em bidirected LP-relaxation}. This is noteworthy for two reasons.  
First, it gives the {\em first concrete evidence} of the
effectiveness of LP-relaxations for minimum latency problems. Second, we show how to
effectively leverage two classes of LPs---bidirected LPs and configuration LPs---that are
often believed to be quite powerful but have only sporadically been effectively leveraged
for network-design and vehicle-routing problems. 
Previously, Chakrabarty and Swamy~\cite{ChakrabartyS11} had proposed some
LP-relaxations (including a configuration LP) for minimum-latency problems 
but they could not improve upon the current-best approximation guarantees for these
problems via these LPs.
Our LPs are inspired by their formulations, and coincide with their LPs in some cases, 
but are subtly stronger, and, importantly (as noted above), our guarantees do indeed
improve the state-of-the-art for 
these problems and testify to the effectiveness of LP-based methods. 

Our algorithms are versatile and extend
easily to handle various extensions involving weighted sum of node latencies, node-depot
service constraints, 
and node service times (Section~\ref{extn}). 

Finally, we propose a configuration LP 
that sheds further light on the power of LP-relaxations for
minimum-latency problems and why they merit further investigation.
We prove that this LP has integrality gap at most $\mu^*<3.5912$ for (the general
setting of) multi-depot \kmlp (see Theorem~\ref{lp2facts}), both 
by devising an efficient rounding procedure, and by showing that this LP is at least as
strong as a combinatorial stroll-based bound that is frequently used for minimum-latency
problems, especially when $k=1$ for which this leads to the current-best approximation
ratio. The latter result implies 
an integrality gap of at most $3.03$ when $k=1$ (due to the analysis of the $\ell$-stroll
lower bound in~\cite{ArcherB10}). 
We do not know how to solve this LP efficiently in general, but we can efficiently obtain a
$(1+\e)$-approximate solution when $k=1$, and thereby efficiently obtain an LP-relative
$(\mu^*+\ve)$-approximation, which essentially matches the current-best
approximation ratio for single-vehicle (i.e., $k=1$) \mlp. 

\medskip
We now sketch the main ideas underlying our algorithms and analyses.
Our $8.497$-approximation for multi-depot \kmlp (Section~\ref{multi}) leverages a natural
configuration LP \eqref{lp1}, where we define a configuration for time $t$ and vehicle $i$
to be an $r_i$-rooted path of length at most $t$.
Using known results on orienteering~\cite{ChaudhuriGRT03} and the arguments
in~\cite{ChakrabartyS11}, one can compute a fractional ``solution'' of cost at most 
$\OPT_\Pc$, 
where (roughly speaking) the solution computed specifies for every $t$ and $i$, a
distribution over $r_i$-rooted {\em trees} (instead of paths) of length roughly $t$, such  
that this ensemble covers nodes to the appropriate extents (Lemma~\ref{lpsolve}). 
The rounding algorithm is then simple: 
we consider time points of geometrically increasing value, we sample
a tree for each time point $t$ and vehicle $i$ from the distribution for $t,i$, convert
this tree to a tour, and finally, for each vehicle $i$, we concatenate the various tours
obtained for $i$ to obtain $i$'s route. 
Compared to the combinatorial algorithm in~\cite{ChekuriK04}, we gain significant savings
from the fact that the LP solution readily yields, for each time $t$, a random $k$-tuple of
trees whose expected coverage can be related to the coverage of the LP solution. 
In contrast, \cite{ChekuriK04} devise an algorithm for the cost version of their variant 
of max-coverage to obtain a $k$-tuple of trees with the desired coverage and
lose various factors in the process. 

The $7.183$-approximation algorithm for single-depot \kmlp (Section~\ref{kmlp}) relies
crucially on {\em bidirected LPs}.
We obtain this guarantee both by rounding a compact bidirected LP-relaxation for the
problem and via more combinatorial arguments where we utilize bidirected LPs to furnish a
key ingredient of the algorithm.
The bidirected LP-relaxation \eqref{lp3} (which is a relaxation for multi-depot \kmlp) works
with the digraph obtained by bidirecting the edges of the input (complete)
graph. The LP specifies the extent to which nodes are covered by each vehicle at each
time, and the extent $z^i_{a,t}$ to which each arc $a$ has been traversed by vehicle $i$'s
route up to time $t$.
We impose suitable node-degree and node-connectivity constraints on the
$z^i_{a,t}$s, and observe that, for each time $t$, one can then use (a polytime version
of) the arborescence-packing results of Bang-Jensen et al.~\cite{BangjensenFJ95} 
(Theorem~\ref{arbpoly}) to decompose $\{\sum_i z^i_{a,t}\}_a$ into a distribution of
$r$-rooted trees with expected length at most $kt$ and covering at least as many nodes as
the LP does by time $t$.
We convert each tree in the support into a $k$-tuple of tours (of length at most
$4t$) 
and stitch together these tours using a concatenation-graph argument similar to the one 
in~\cite{GoemansK98} (losing a $\frac{\mu^*}{2}$-factor), which also shows that the fact
that we have a distribution of trees for each $t$ instead of a single trees does not cause
any problems.  
Theorem~\ref{arbpoly} is precisely what leads to our improvement
over~\cite{FakcharoenpholHR07}: we {\em match} the coverage of an optimal LP-solution at
each step (incurring a certain blow-up in cost),  
whereas~\cite{FakcharoenpholHR07} sequentially find $k$ tours, causing them to lag behind
in coverage and incur a corresponding loss in approximation. 

The combinatorial arguments rely on the following result that is of
independent interest. 
We show that one can efficiently compute an $r$-rooted tree that is at least as good,
with respect to a prize-collecting objective that incorporates both the edge cost and the
number of nodes covered, as the best collection of (any number of, and perhaps non-simple) 
$r$-rooted paths (Theorem~\ref{pcstroll}). 
We obtain this by formulating a bidirected LP for the prize-collecting problem of
finding the desired collection of paths, and rounding it using a polytime version (that we 
prove) of the arborescence-packing results of~\cite{BangjensenFJ95} for weighted digraphs.  
Theorem~\ref{pcstroll} also implies that for every $\ell$, one can efficiently compute an
$r$-rooted tree, or a distribution over two $r$-rooted trees, which we call a 
{\em bipoint tree}, that, in expectation, spans at least $\ell$ nodes and has cost at most
the minimum total cost of a collection of $r$-rooted paths spanning $\ell$ nodes
(Corollary~\ref{pccor}).   
Again, this is where we improve over~\cite{FakcharoenpholHR07} since we match the coverage
of an optimal (integer) solution at each step.  
We compute these objects for all values of $\ell$, convert each constituent tree 
into $k$ tours, and stitch them together as before.

Theorem~\ref{pcstroll} relating prize-collecting trees and path-collections
substantially generalizes a result of Chaudhuri et al.~\cite{ChaudhuriGRT03},  
who prove an analogous result for the special case where one compares (the computed tree) 
against the best {\em single path}. 
Whereas this 
suffices for {\em single-vehicle} \mlp 
(and allowed~\cite{ChaudhuriGRT03} to improve the
approximation for single-vehicle \mlp), 
it 
does not help in  multiple-vehicle settings. (This is 
because to lower bound the $\ell$-th smallest latency incurred by the optimal solution,
one needs to consider a collection of {\em $k$ paths} that together cover $\ell$ nodes.)
Notably, our proof of our more general result is significantly simpler and cleaner (and
different) than the one in~\cite{ChaudhuriGRT03}.
We remark that the approach in~\cite{ChaudhuriGRT03} (for $k= 1$), where one ``guesses''
the endpoint of an optimum path, 
is computationally infeasible for large $k$.

\vspace{-1ex}
\paragraph{Related work.}
Although single-vehicle \mlp (which we refer to simply as \mlp) has attracted much
attention in the Computer Science and Operations Research communities, 
there is little prior work on multi-vehicle versions of \mlp, especially from the
perspective of approximation algorithms.  
Chekuri and Kumar~\cite{ChekuriK04}, and Fakcharoenphol et al.~\cite{FakcharoenpholHR03}  
seem to be the first ones to consider multi- and single- depot \kmlp
respectively; they obtained approximation ratios of $12\beta$ and $8.497\beta$ for these
problems respectively, where $\beta$ is the approximation ratio of the \lmst problem.%
\footnote{The $12\beta$-approximation in~\cite{ChekuriK04} is a consequence of the
algorithm they develop for the cost-version of the max-coverage variant that they
consider.  
Although the cardinality-version of their problem now admits a better approximation ratio
since it can be cast as a submodular-function maximization problem subject to a matroid
constraint~\cite{CalinescuCPV11}, this improvement does not apply to the cost-version, which
gives rise to multiple knapsack constraints.}
Subsequently, a result of Chaudhuri et al.~\cite{ChaudhuriGRT03} relating
prize-collecting trees and paths provided a general tool that allows one eliminate the
$\beta$ term in the above approximation ratios. 
Recently Sitters~\cite{Sitters14} obtained a PTAS for \mlp 
on trees, the Euclidean plane, and planar graphs, and mentions that the underlying
techniques extend to yield a PTAS on these graphs for ({\em single-depot}) \kmlp for any 
{\em constant $k$}. We are not aware of any other work on \kmlp.

We now discuss some relevant work on \mlp (i.e., $k=1$). 
\mlp (and hence \kmlp) is known to be 
hard to approximate to better than some constant factor~\cite{BC+94,PapadimitriouY93}, and
\nphard even on trees~\cite{Sitters02}.
While much work in the Operations Research literature has focused on exactly solving \mlp (in
exponential time)~\cite{Lucena90,SimchileviB91,FischettiLM93,BiancoMR93},
Blum et al.~\cite{BC+94}
devised the first constant-factor approximation for \mlp via the process of finding
tours of suitable lengths and/or node coverages and concatenating them. They obtained a
$\min\{144,8\beta\}$-approximation. 
Subsequently, \cite{GoemansK98} refined the
concatenation procedure in~\cite{BC+94} and proposed the device of a concatenation graph
to analyze this process, which yielded an improved $\mu^*\beta$-approximation, where
$\mu^*<3.5912$ is the solution to $\mu\ln\mu=1+\mu$.
The procedure of stitching together tours and its analysis via the concatenation graph
have since become standard tools in the study of minimum-latency problems.
Archer et al.~\cite{ArcherLW08} showed that one can replace the \lmst-subroutine
in the algorithm of~\cite{GoemansK98} by a so-called 
{\em Lagrangian-multiplier preserving} (LMP) $\beta'$-approximation algorithm for the
related {\em prize-collecting Steiner tree} (\pcst) problem, and
thereby achieve a $2\mu^*$-approximation using the LMP $2$-approximation for
\pcst~\cite{GoemansW95}. The current-best approximation 
for \mlp is due to Chaudhuri et al.~\cite{ChaudhuriGRT03} who 
showed that the factors $\beta$ and $\beta'$ above can be eliminated, 
leading to a $\mu^*$-approximation, by noting that:
(i) the lower bound $\sum_{\ell=1}^n(\text{optimal value of \lmst})$ used in all previous
works starting with~\cite{BC+94} 
can be strengthened to the {\em $\ell$-stroll lower bound} by replacing the summand with the
optimal cost of a rooted {\em path} covering $\ell$ nodes; and 
(ii) one can adapt the arguments in~\cite{GoemansW95,Garg96} to obtain a prize-collecting
tree of cost no more than that of an optimal prize-collecting tree. 
As noted earlier, (ii) is a rather special case of our Theorem~\ref{pcstroll}. 

Chakrabarty and Swamy~\cite{ChakrabartyS11} 
proposed some LP relaxations for minimum-latency problems 
and suggested that their LPs may lead to improvements for these problems. 
This was the inspiration for our work. Our LPs are subtly different, 
but our work lends credence to the idea that LP-relaxations for minimum-latency
problems can lead to improved guarantees for these problems.

Improved guarantees are known for \mlp in various special cases. Arora and
Karakostas~\cite{AroraK03} give a quasi-PTAS for trees, and Euclidean metrics in any finite 
dimension. Sitters~\cite{Sitters14} recently improved these to a PTAS for trees, 
the Euclidean plane, and planar graphs.

\section{Preliminaries} \label{prelim}
Recall that in the {\em multi-depot $k$-vehicle minimum latency problem} (multi-depot
\kmlp), we have a complete undirected graph $G=(V,E)$ on $n$ nodes, metric edge costs
$\{c_e\}$, and a set $R=\{r_1,\ldots,r_k\}$ of $k$ root/depot nodes.
The goal is to find $k$ paths $P_1,\ldots,P_k$, where each path $P_i$ starts at $r_i$, 
so that $\bigcup_{i=1}^k V(P_i)=V$, 
and the total latency $\sum_{i=1}^k\sum_{v\in P_i:v\neq r_i}c_{P_i}(v)$ incurred is minimized. 
We call the special case where $r_1=\ldots=r_k$ 
{\em single-depot \kmlp} and abbreviate this to \kmlp.
We sometimes refer to non-root nodes as clients. 
We may assume that the $c_e$s are integers, and $c_{r_iv}\geq 1$ for every non-depot
node $v$ and every root $r_i$.

Algorithms for minimum-latency problems frequently use
the idea of concatenating tours to build a solution, and our algorithms also follow this
general template. 
A {\em concatenation graph}~\cite{GoemansK98} is a convenient means of representing this 
concatenation process. 
The concatenation graph corresponding to a sequence $C_1=0,\ldots,C_n$ of nonnegative
numbers (such as the lengths of tours spanning $1, 2, \ldots, n$ nodes) denoted
$\cg(C_1,\ldots,C_n)$, is a directed graph with $n$ nodes, and an arc 
$(i,j)$ of length $C_j\bigl(n-\frac{i+j}{2}\bigr)$ for all $i<j$.
We collect below some useful facts 
about this graph.
We say that $C_\ell$ is an {\em extreme point} of the sequence $(C_1,\ldots,C_n)$ if
$(\ell,C_\ell)$ is extreme-point of the convex hull of $\{(j,C_j): j=1,\ldots,n\}$.

\begin{theorem}[~\cite{GoemansK98,ArcherLW08,ArcherB10}] \label{cgthm} \label{thm_concat_graph}
The shortest $1\leadsto n$ path in $\cg(C_1,\ldots,C_n)$ has length at most
$\frac{\mu^*}{2}\sum_{\ell=1}^n C_\ell$, 
where $\mu^*<3.5912$ is the solution to $\mu\ln\mu=\mu+1$.
Moreover, the shortest path only visits nodes
corresponding to extreme points of $(C_1,\ldots,C_n)$. 
\end{theorem}

Given a point-set $S\sse\R_+^2$, define its {\em lower-envelope curve} 
$f:[\min_{(x,y)\in S}x,\max_{(x,y)\in S}x]\mapsto\R_+$  
by $f(x)=\min\{y: y\in\conv(S)\}$, where $\conv(S)$ denotes the convex hull of $S$.
Note that $f$ is well defined since the minimum is taken over a closed, compact set.

Let $f$ be the lower-envelope curve of $\{(j,C_j): j=1,\ldots,n\}$, where $C_1=0$.
If $C_\ell$ is an extreme point of $(C_1,\ldots,C_n)$, we will often say that
$(\ell,C_\ell)$ is a ``corner point'' of $f$.
Notice that the bound in Theorem~\ref{cgthm} on the shortest-path length can be
strengthened to $\frac{\mu^*}{2}\sum_{\ell=1}^n f(\ell)$. This is because the shortest path
$P_f$ in the concatenation graph $\cg(f(1),\ldots,f(n))$ has length at most
$\frac{\mu^*}{2}\sum_{\ell=1}^n f(\ell)$, and uses only extreme points of
$\bigl(f(1),\ldots,f(n)\bigr)$, which in turn must be extreme points of $(C_1,\ldots,C_n)$
since $f$ is the lower-envelope curve of $\{(j,C_j): j=1,\ldots,n\}$. Hence, $P_f$ is also
a valid path in $\cg(C_1,\ldots,C_n)$, and its length in these two graphs is exactly the
same. Corollary~\ref{cgcor} shows that this bound can further be strengthened to
$\frac{\mu^*}{2}\int_1^nf(X)dx$. This will be useful in Section~\ref{kmlp}.
The proof is similar to the above argument and follows by discretizing $f$ using finer and
finer scales; we defer the proof to Appendix~\ref{append-prelim}.

\begin{corollary} \label{cgcor}
The shortest $1\leadsto n$ path in $\cg(C_1,\ldots,C_n)$ has length at most
$\frac{\mu^*}{2}\int_{1}^n f(x)dx$, where $f:[1,\ldots,n]\mapsto\R_+$ is the
lower-envelope curve of $\{(j,C_j): j=1,\ldots,n\}$, and only visits nodes
corresponding to extreme points of $(C_1,\ldots,C_n)$. 
\end{corollary}

\paragraph{The bottleneck-stroll lower bound.}
Our algorithms for single-depot \kmlp utilize a combinatorial stroll-based lower bound
that we call the {\em $(k,\ell)$-bottleneck-stroll lower bound} (that applies even to
multi-depot \kmlp), denoted by $\bnslb$, which is obtained as follows. 
Given an instance $\bigl(G=(V,E),\{c_e\},k,R=\{r_1,\ldots,r_k\}\bigr)$ of multi-depot
\kmlp, in the {\em $(k,\ell)$-bottleneck-stroll} problem, 
we seek $k$ paths $P_1,\ldots,P_k$, where each $P_i$ is rooted at $r_i$, that together
cover at least $\ell$ nodes (that may include root nodes) so as to minimize $\max_i c(P_i)$. 
Let $\optbns{k,\ell}$ denote the cost of an optimal solution.
It is easy to see that if $t^*_\ell$ is the $\ell$-th smallest node latency incurred by an
optimal solution, then $t^*_\ell\geq\optbns{k,\ell}$. We define 
$\bnslb:=\sum_{\ell=1}^{|V|}\optbns{k,\ell}$, which is clearly a lower bound on the
optimum value of the multi-depot \kmlp instance. 

To put this lower bound in perspective, we remark that a concatenation-graph argument
dovetailing the one used for \mlp in~\cite{GoemansK98} shows that there is
multi-depot-\kmlp solution of cost at most $\mu^*\cdot\bnslb$ (see Theorem~\ref{bnsbnd}). 
For $k=1$, $\bnslb$ becomes the $\ell$-stroll lower bound   
in~\cite{ChaudhuriGRT03}, which yields the current-best approximation factor for \mlp.
Also, the analysis in~\cite{ArcherB10} shows that there is an \mlp-solution of cost at
most $3.03\cdot\bnslb$. On the other hand, whereas such combinatorial stroll-based lower
bounds have been frequently leveraged for minimum-latency problems, we provide some
evidence in Section~\ref{altlp} that LPs may prove to be even more powerful by describing
a configuration LP whose optimal value is always at least \bnslb.
We defer the proof of the following theorem to Appendix~\ref{append-bns}. 

\begin{theorem} \label{bnsbnd} 
(i) There is a solution to multi-depot \kmlp of cost at most $\mu^*\cdot\bnslb$. 
(ii) There is a solution to \mlp of cost at most $3.03\cdot\bnslb$.  
\end{theorem}

\section{Arborescence packing and prize-collecting trees and paths} \label{arbpack}
A key component of our algorithms for (single-depot) \kmlp is a polytime version of an
arborescence-packing result of~\cite{BangjensenFJ95} that we prove for weighted digraphs
(Theorem~\ref{arbpoly}).  
We use this to obtain a result relating trees and paths that we utilize in our
``combinatorial'' algorithm for \kmlp, which we believe is of independent interest. Let $r$ be
a given root node. Throughout this section, when we say rooted tree or rooted path, we
mean a tree or path rooted at $r$.
We show that for any $\ld>0$, one can efficiently find a rooted tree $T$ such that
$c(T)+\ld|V\sm V(T)|$ is at most $\sum_{P\in\C}c(P)+\ld|V\sm\bigcup_{P\in\C}V(P)|$, 
where $\C$ is any collection of rooted paths (see Theorem~\ref{pcstroll}).
As noted earlier, 
this substantially generalizes a result in~\cite{ChaudhuriGRT03},  
yet our proof is simpler. 

For a digraph $D$ (possibly with parallel edges), we use $\ld_D(x,y)$ to denote the
number of $x\leadsto y$ edge-disjoint paths in $D$.
Given a digraph $D$ with nonnegative integer edge weights $\{w_e\}$, we define the
quantities $|\dt^\into(u)|$, $|\dt^\out(u)|$ and $\ld_D(x,y)$ for $D$ to be the respective
quantities for the unweighted (multi-)digraph obtained by replacing each edge $e$ of $D$
with $w_e$ parallel edges (that is, $|\dt^\into(u)|=\sum_{e=(\circ,u)}w_e$ etc.).
Bang-Jensen at al.~\cite{BangjensenFJ95} proved an arborescence-packing result that in
particular implies that if $D = (U+r,A)$ is a digraph with root $r\notin U$ such that
$|\delta^{\into}(u)| \ge |\delta^{\out}(u)|$ for all $u \in U$,  
then, for any integer $k\geq 0$, one can find $k$ edge-disjoint out-arborescences rooted
at $r$ such that every node $u\in U$ belongs to at least $\min\{k,\lambda_D(r,u)\}$
arborescences. 
For a weighted digraph, applying this result on the corresponding
unweighted digraph 
yields a pseudopolynomial-time algorithm for finding the stated arborescence family.
We prove the following {\em polytime} version of their result 
for weighted digraphs; 
the proof appears in Appendix~\ref{append-arbpoly}. 

\begin{theorem} \label{cor_poly_arb_packing} \label{arbpoly}
Let $D=(U+r,A)$ be a digraph with nonnegative integer edge weights $\{w_e\}$, 
where $r\notin U$ is a root node, such that $|\delta^{\into}(u)| \ge |\delta^{\out}(u)|$
for all $u \in U$. 
For any integer $K\geq 0$, one can find out-arborescences $F_1,\ldots,F_q$
rooted at $r$ and integer weights $\gm_1,\ldots,\gm_q$ in polynomial time such that  
$\sum_{i=1}^q\gm_i=K$, $\sum_{i:e\in F_i}\gm_i\leq w_e$ for all $e\in A$, 
and $\sum_{i:u\in F_i}\gm_i=\min\{K,\ld_D(r,u)\}$ for all $u\in U$.
\end{theorem}

We now apply Theorem~\ref{arbpoly} to prove our key result relating prize-collecting
arborescences and prize-collecting paths. Let $G=(V,E)$ be a complete undirected graph
with root $r\in V$, metric edge costs $\{c_e\}$, and nonnegative node penalties
$\{\pi_v\}_{v\in V}$.  
We bidirect the edges to obtain a digraph $D=(V,A)$, setting the cost of both $(u,v)$ and
$(v,u)$ to $c_{uv}$.
Consider the following bidirected LP-relaxation for the problem of finding a collection
$\C$ of rooted paths minimizing 
$\sum_{P\in\C}c(P)+\pi\bigl(V\sm\bigcup_{P\in\C}V(P)\bigr)$.
We use $a$ to index edges in $A$, and $v$ to index nodes in $V\sm\{r\}$.
\begin{alignat*}{3}
\min & \quad & \sum_a c_ax_a+\sum_v\pi_vz_v \qquad 
\text{s.t.} \qquad
x\bigl(\dt^{\into}(v)\bigr) & \geq x\bigl(\dt^{\out}(v)\bigr) 
\quad && \forall v\in V\sm\{r\} \tag{PC-LP} \label{bns_lp} \\
&& x\bigl(\dt^{\into}(S)\bigr)+z_v & \geq 1 \quad && \forall S\sse V\sm\{r\}, v\in S;
\qquad x,z \geq 0.
\end{alignat*}

\newcommand{\optpclp}{\OPT_{\text{\ref{bns_lp}}}}

\begin{theorem} \label{pcstroll}
(i) We can efficiently compute a rooted tree $T$ such that 
$c(T)+\pi(V\sm V(T))\leq\optpclp$. 

\noindent 
(ii) Hence, for any $\ld\geq 0$, we can find a tree $T_\ld$ such that 
$c(T_\ld)+\ld|V\sm V(T_\ld)|\leq\sum_{P\in\C}c(P)+\ld|V\sm\bigcup_{P\in\C}V(P)|$ for any 
collection $\C$ of rooted paths. 
\end{theorem}

\begin{proof}
Let $(x,z)$ be an optimal solution to \eqref{bns_lp}. Let $K$ be such that $Kx_a$ is an
integer for all $a$; note that $\log K$ is polynomially bounded in the input
size. 
Consider the digraph $D$ with edge weights $\{Kx_a\}$. 
Let $(\gm_1,F_1),\ldots,(\gm_q,F_q)$ be the weighted arborescence family obtained by
applying Theorem~\ref{arbpoly} to $D$ with the integer $K$.
Then, we have: 
(a) $\sum_{i=1}^q\gm_i=K$;
(b) $\sum_{i=1}^q\gm_ic(F_i)=\sum_ac_a\bigl(\sum_{i:a\in F_i}\gm_i\bigr)\leq K\sum_ac_ax_a$;
and 
(c) $\sum_{i=1}^q\gm_i\pi\bigl(V\sm V(F_i)\bigr)=
\sum_v\pi_v\bigl(\sum_{i:v\notin F_i}\gm_i\bigr)\leq K\sum_v\pi_vz_v$, 
where the last inequality follows since $\ld_D(r,v)\geq K(1-z_v)$ for all $v\in V\sm\{r\}$.
Thus, if we take the arborescence $F_i$ with minimum prize-collecting
objective $c(F_i)+\pi(V\sm V(F_i))$, which we will treat as a rooted tree $T$ in $G$, we
have that $c(T)+\pi(V\sm V(T))\leq\sum_ac_ax_a+\sum_v\pi_vz_v$.

For part (ii), let $T_\ld$ be the tree obtained in part (i) with penalties $\pi_v=\ld$ for all
$v\in V$. Observe that any collection $\C$ of rooted paths yields a feasible solution to
\eqref{bns_lp} of cost $\sum_{P\in\C}c(P)+\ld|V\sm\bigcup_{P\in\C}V(P)|$.

Notice that the above proof does not use the fact that the edge costs are symmetric or
form a metric; so the theorem statement holds with {\em arbitrary} nonnegative edge costs
$\{c_a\}_{a\in A}$. 
\end{proof}

Given Theorem~\ref{pcstroll} for the prize-collecting problem, one can use binary search
on the 
parameter $\ld$ to obtain the following result for the partial-cover
version; the proof appears in Appendix~\ref{append-pccor}.  
A rooted {\em bipoint tree} $\bT=(a,T_1,b,T_2)$, where $a,b\geq0,\ a+b=1$, is a convex combination 
$aT_1+bT_2$ of two rooted trees $T_1$ and $T_2$. 
We extend a function $f$ defined on trees to bipoint trees by setting
$f(\bT)=af(T_1)+bf(T_2)$.

\begin{corollary} \label{pccor}
(i) Let $B\geq 0$, and $\iopt$ be the minimum cost of a collection
of rooted paths spanning at least $B$ nodes. We can efficiently compute a rooted tree or
bipoint tree $Q$ such that $c(Q)\leq\iopt$
and $|V(Q)|=B$.

\noindent
(ii) Let $\{w_v\}$ be nonnegative node penalties with $w_r=0$. Let $C\geq 0$, and
$n^*$ be the maximum node weight of a
collection of rooted paths of total cost at most $C$. We
can efficiently compute a rooted tree or bipoint tree $Q$ such that $c(Q)=C$ and
$w(V(Q))\geq n^*$.
\end{corollary}

\section{\boldmath LP-relaxations for multi-depot \kmlp} 
\label{lps} 
Our LP-relaxations are time-indexed formulations inspired by 
the LPs in~\cite{ChakrabartyS11}. Let $\lb:=\max_v\min_i c_{r_iv}$. 
Let $\Time\leq 2n\lb$ be an upper bound on the maximum latency of a node that can be
certified by an efficiently-computable solution. 
Standard scaling and rounding can be used to ensure that
$\lb=\poly\bigl(\frac{n}{\e}\bigr)$ at the expense of a $(1+\e)$-factor loss
(see, e.g.,~\cite{AroraK03}). So we assume in the sequel that $\Time$ is polynomially
bounded. In Section~\ref{extn}, we sketch an approach showing how to solve our
time-indexed LPs without this assumption, which turns out to be useful for some of the
extensions that we consider.
In either case, this means that all our LP-based guarantees degrade by a $(1+\e)$
multiplicative factor. 
Throughout, we use $v$ to index the non-root nodes in $V\sm R$, $i$ to index the $k$
vehicles, and $t$ to index time units in $[\Time]:=\{1,2,\ldots,\Time\}$. 

In Section~\ref{config}, we describe two {\em configuration LPs} with
exponentially many variables. The first LP, \eqref{lp1}, can be ``solved'' efficiently and
leads to an $8.497$-approximation for multi-depot \kmlp (Section~\ref{multi}). The second
LP, \eqref{lp2}, is a stronger LP that we do not know how to solve efficiently (except
when $k=1$), but whose integrality gap is much smaller (see Theorem~\ref{lp2facts}). In
Section~\ref{bidirect}, we describe a {\em bidirected LP-relaxation} with exponentially
many cut constraints that one can separate over and hence solve the LP (assuming $\Time$
is polynomially bounded). This LP is weaker than \eqref{lp1}, but we show in
Section~\ref{kmlp} that this leads to a $7.813$-approximation algorithm for \kmlp and a
$\mu^*$-approximation algorithm for \mlp (i.e., $k=1$). 
Theorem~\ref{lprelate} summarizes the relationship between the various LPs and the
guarantees we obtain relative to these via the algorithms described in the following
sections. 

\subsection{Configuration LPs} \label{config}
The idea behind a configuration LP is to have variables for each time $t$ describing the 
snapshot of the vehicles' routes up to time $t$. Different LPs arise depending on whether
the snapshot is taken for each individual vehicle, or is a global snapshot of the $k$
vehicles' routes.  

Let $\Pcol^i_t$ and $\Tcol^i_t$ denote respectively the collection of all (simple) paths and
trees rooted at $r_i$ of length at most $t$. 
In our first configuration LP, we introduce a variable $z^i_{P,t}$ for every time $t$ and path
$P\in\Pcol^i_t$ that indicates if $P$ is the path used to visit the nodes on vehicle $i$'s
route having latency at most $t$; that is, $z^i_{P,t}$ denotes if $P$ is the
portion of vehicle $i$'s route up to time $t$.  
We also have variables $x^i_{v,t}$ to denote if node $v$ is visited at time $t$ by the
route originating at root $r_i$. 

\vspace{-10pt}
\noindent \hspace*{-3ex}
\begin{minipage}[t]{.49\textwidth}
\begin{alignat}{3}
\min & \quad & \sum_{v,t,i}tx^i_{v,t} & \tag{LP$_\Pcol$} \label{lp1} \\
\text{s.t.} && \sum_{t,i} x^i_{v,t} & \geq 1 \qquad && \forall v \label{jasgn} \\
&& \sum_{P\in\Pcol^i_t} z_{P,t} & \leq 1 \qquad && \forall t,i \label{onep} \\ 
&& \sum_{P\in\Pcol^i_t: v\in P}z^i_{P,t} & \geq \sum_{t'\leq t}x^i_{v,t'} 
\qquad && \forall v,t,i \label{jcov} \\
&& x, z & \geq 0. \notag 
\end{alignat}
\end{minipage}
\quad \rule[-29ex]{1pt}{26ex}
\begin{minipage}[t]{.5\textwidth}
\begin{alignat}{3}
\max & \quad & \sum_v\al_v & - \sum_{t,i}\beta^i_t \tag{D} \label{dlp1} \\
\text{s.t.} && \al_v \leq t & +\sum_{t'\geq t}\tht^i_{v,t'} \qquad && \forall v,t,i
\label{dineq1} \\
&& \sum_{v\in P}\tht^i_{v,t} & \leq \beta^i_t \qquad && \forall t, i, P\in\Pcol_t 
\label{dineq2} \\ 
&& \al, \beta, \tht & \geq 0. \label{dineq3}
\end{alignat}
\end{minipage}

\medskip \noindent
Constraint \eqref{jasgn} encodes that every non-root node must be visited by some vehicle
at some time; \eqref{onep} and \eqref{jcov} encode that at most one path corresponds to
the portion of vehicle $i$'s route up to time $t$, and that this path must visit every
node $v$ visited at any time $t'\le t$ by vehicle $i$.
Note that these enforce that $x^i_{v,t}=0$ if $t<c_{r_iv}$.
We remark that in the single-depot case, the splitting of the $k$ paths into one path per
vehicle is immaterial, and so \eqref{lp1} becomes equivalent to the configuration LP 
in~\cite{ChakrabartyS11} for \kmlp that involves a single set of $x_{v,t}$ and $z_{P,t}$
variables for each time $t$. 

In order to solve \eqref{lp1}, we consider the dual LP \eqref{dlp1}, which has
exponentially many constraints. Separating over constraints \eqref{dineq2} involves
solving a (rooted) path-orienteering problem: for every $t$, given rewards
$\{\tht^i_{u,t}\}_{u\in V}$, where we set $\tht^i_{u,t}=0$ for $u\in R$, we want to
determine if there is a path $P$ rooted at $r_i$ of length at most $t$ that gathers reward
more than $\beta^i_t$. 
In unweighted orienteering, all node rewards are $0$ or $1$.
A $(\rho,\gm)$-\{path, tree\} approximation algorithm for the path-orienteering problem is
an algorithm that always returns a \{path, tree\} rooted at $r_i$ of length at most
$\gm(\text{length bound})$ that gathers reward at least $(\text{optimum reward})/\rho$. 

As shown in~\cite{ChakrabartyS11}, one can use such approximation algorithms to obtain an
approximate solution to \eqref{lp1} (and similar configuration LPs), where the notion
of approximation involves bounded violation of the constraints and moving to a ``tree
version'' of \eqref{lp1}.  
In the tree version of a configuration LP such as \eqref{lp1}, the only
change is that configurations are defined in terms of trees instead of paths. 
Specifically, define the tree-version of \eqref{lp1}, denoted $(\text{LP}_\Tcol)$, to be
the analogue where we have variables $z^i_{Q,t}$ for every $Q\in\Tcol^i_t$, and we replace
all occurrences of $z^i_{P,t}$ in \eqref{lp1} with $z^i_{Q,t}$. 

\newcommand{\mlpath}[1]{\ensuremath{\bigl(\text{LP}_{\Pcol}^{({#1})}\bigr)}}
\newcommand{\mltree}[1]{\ensuremath{\bigl(\text{LP}_{\Tcol}^{({#1})}\bigr)}}

Let $\mlpath{a}$ be \eqref{lp1} where
we replace each occurrence of $\Pcol^i_t$ in constraints \eqref{onep}, \eqref{jcov} by
$\Pcol^i_{at}$, and the RHS of \eqref{onep} is now $a$. Let $\mltree{a}$ be defined
analogously. Let $\OPT_\Pc$ be the optimal value of \eqref{lp1} (i.e., $\mlpath{1}$).  
Chaudhuri et al.~\cite{ChaudhuriGRT03} give a $(1,1+\e)$-tree approximation for unweighted  
orienteering, which yields (via suitably scaling and rounding the node rewards) a
$(1+\e,1+\e)$-tree approximation for weighted orienteering. Utilizing this and
mimicking the arguments in~\cite{ChakrabartyS11} yields
Lemma~\ref{lpsolve}, which combined with our rounding procedure in Section~\ref{multi}
yields an $(8.497+\ve)$-approximation algorithm for multi-depot \kmlp
(Theorem~\ref{multithm}). 

\begin{lemma} \label{lpsolve}
For any $\e>0$, we can compute a feasible solution to $\mltree{1+\e}$ of cost at most 
$\OPT_\Pc$ in time $\poly\bigl(\text{input size},\frac{1}{\e}\bigr)$. 
\end{lemma}

\newcommand{\noptlp}{\ensuremath{\OPT_{\text{\ref{lp2}}}}}

\paragraph{A stronger configuration LP.} 
We now describe a stronger LP \eqref{lp2} that sheds further light on the power of
LP-relaxations for minimum-latency problems. We prove that the integrality gap of
\eqref{lp2} is at most $\mu^*<3.5912$ by giving an efficient rounding procedure
(Theorem~\ref{lp2facts} (ii)). 
We do not know how to leverage this to obtain 
an efficient $\mu^*$-approximation for multi-depot \kmlp, since we do not know 
how to solve \eqref{lp2} efficiently, even in the approximate sense of Lemma~\ref{lpsolve}.
But for $k=1$, \eqref{lp2} coincides with \eqref{lp1} (and the configuration LP
in~\cite{ChakrabartyS11}), so we can use Lemma~\ref{lpsolve} to approximately solve
\eqref{lp2}. 
We also show that $\noptlp$ is at least the value of 
the $(k,\ell)$-bottleneck-stroll lower bound, $\bnslb$. 
Combined with Theorem~\ref{bnsbnd}, this provides another proof that the integrality gap
of \eqref{lp2} is at most $\mu^*$; this also shows that the integrality gap of \eqref{lp2} 
\nolinebreak \mbox{is at most $3.03$ when $k=1$.}

In the new LP, a configuration for time $t$ is the global snapshot of the $k$ vehicles'
routes up to time $t$; that is, 
it is the $k$-tuple formed by the portions of the $k$ vehicles' routes up to time $t$. 
Formally, a configuration $\vP$ for time $t$ 
is a tuple $(P_1,\ldots,P_k)$ where each $P_i$ is rooted at $r_i$ and has length at most
$t$. We say that $v$ is covered by $\vP$, and denote this by $v\in\vP$, to mean that
$v\in\bigcup_i V(P_i)$. 
Let $\vpcol_t$ denote the collection of all configurations for time $t$. 
This yields the following LP whose 
constraints encode that every
non-root node must be covered, there is at most one configuration for each time $t$, and
this configuration must cover every node $v$ whose latency is at most $t$.
\begin{alignat}{3}
\min & \quad & \sum_{v,t}tx_{v,t} & \tag{LP2$_\vpcol$} \label{lp2} \\
\text{s.t.} && \sum_{t} x_{v,t} & \geq 1 \qquad && \forall v \label{jasgnconfig} \\
&& \sum_{\vP\in\vpcol_t} z_{\vP,t} & \leq 1 \qquad && \forall t \label{oneconfig} \\ 
&& \sum_{\vP\in\vpcol_t: v\in\vP}z_{\vP,t} & \geq \sum_{t'\leq t}x_{v,t'} 
\qquad && \forall v,t \label{jcovconfig} \\
&& x, z & \geq 0. \notag 
\end{alignat}

As before, we may define a tree version
of \eqref{lp2} 
similarly to the way in which $(\text{LP}_\Tcol)$ is obtained from \eqref{lp1}. 
Define a {\em tree configuration} $\vQ$ for time $t$ to be a tuple
$(Q_1,\ldots,Q_k)$, where each $Q_i$ is an $r_i$-rooted tree of cost at most $t$. As
before, we say that $v$ is covered by $\vQ$ if $v\in\bigcup_i V(Q_i)$ denote this by
$v\in\vQ$. Let $\vtcol_t$ denote the collection of all tree configurations for time $t$.  
In the tree-version of \eqref{lp2}, denoted $(\text{LP2}_\vtcol)$, 
we have variables $z_{\vQ,t}$ for every $\vQ\in\vtcol_t$, and 
we replace constraints \eqref{oneconfig}, \eqref{jcovconfig} by

\noindent
\begin{minipage}[t]{.49\textwidth}
\begin{equation}
\sum_{\vQ\in\vtcol_{t}} z_{\vQ,t} \leq 1 \qquad \forall t \label{onetree} 
\end{equation}
\end{minipage} \quad
\begin{minipage}[t]{.49\textwidth}
\begin{equation}
\sum_{\vQ\in\vtcol_{t}: v\in\vQ}z_{\vQ,t} \geq \sum_{t'\leq t}x_{v,t'} 
\qquad \forall v,t \label{jcovtree}
\end{equation}
\end{minipage}

\newcommand{\configtree}[1]{\ensuremath{\bigl(\text{LP2}_{\vtcol}^{({#1})}\bigr)}}

\medskip \noindent
We define $\configtree{a}$ to be $(\text{LP2}_\vtcol)$, where we replace $\vtcol_t$ in
\eqref{onetree}, \eqref{jcovtree} with $\vtcol_{at}$, and we replace the RHS of
\eqref{onetree} with $a$.
The following theorem suggests that \eqref{lp2} may be quite powerful; we defer its proof
to Section~\ref{altlp}.

\begin{theorem} \label{lp2facts}
We have the following.
\begin{list}{{(\roman{enumi})}}{\usecounter{enumi} \topsep=0.5ex \itemsep=0ex
    \addtolength{\leftmargin}{0ex} \addtolength{\labelwidth}{\widthof{(iii)}}}
\item $\noptlp\geq\bnslb$ for every instance of multi-depot \kmlp.
\item A solution to \eqref{lp2}, can be efficiently rounded to a feasible integer solution
while increasing the cost by a factor of at most $\mu^*<3.5912$.  
Furthermore, for any $\e\geq 0$, a solution to $\configtree{1+\e}$ can be efficiently
rounded to a feasible solution to multi-depot \kmlp while losing a factor of at most   
$\frac{\mu^*}{1-\mu^*\e}$.
\item When $k=1$, for any $\e>0$, we can compute a feasible solution to
$\configtree{1+\e}$ of cost at most $\noptlp$ in time 
$\poly\bigl(\text{input size},\frac{1}{\e}\bigr)$. 
\end{list}
\end{theorem}

\subsection{A bidirected LP relaxation} \label{bidirect}
The bidirected LP formulation is motivated by Theorem~\ref{arbpoly}. As in
Section~\ref{arbpack}, we bidirect the edges to obtain a digraph $D=(V,A)$ and set
$c_a=c_{uv}$ for both $a=(u,v)$ and $a=(v,u)$. 
We use $a$ to index the arcs in $A$.
Recall that $v$ indexes nodes in $V\sm R$, $i$ indexes the $k$ vehicles, and $t$ indexes
time units in $[\Time]$. 
As before, we use variables $x^i_{v,t}$ to denote if node $v$ is visited at time $t$ by
the route originating at root $r_i$. Directing the vehicles' routes away from their roots
in a solution, $z^i_{a,t}$ indicates if arc $a$ lies on the portion of vehicle $i$'s route
up to time $t$. We obtain the following LP.
\begin{alignat}{3}
\min & \quad & \sum_{v,t,i} tx^i_{v,t} & \tag{LP3} \label{lp3} \\
\text{s.t.} && \sum_{t,i} x^i_{v,t} & \ge 1 \qquad && \forall v; 
\qquad x^i_{v,t}=0\ \text{if $c_{r_iv}<t$} \qquad \forall v,t,i \label{jasgn3} \\
&& \sum_{a\in\dt^{\into}(S)}z^i_{a,t} & \ge \sum_{t'\le t} x^i_{v,t'} 
\qquad && \forall S\subseteq V\sm\{r_i\}, v\in S, \quad \forall t \label{jcov3} \\
&& \sum_{a} c_az^i_{a,t} & \le t \qquad && \forall t,i \label{onep3} \\
&& \sum_{a\in\dt^{\into}(v)}z^i_{a,t} & \geq\sum_{a\in\dt^{out}(v)}z^i_{a,t} 
\qquad && \forall v, i \label{path3} \\
&& x, z & \geq 0.
\end{alignat}
Constraints \eqref{jasgn3} ensure that every non-root node is visited at some time.
Constraints \eqref{jcov3}--\eqref{path3} play the role of constraints \eqref{onep},
\eqref{jcov} in \eqref{lp1}: \eqref{jcov3} ensures that the portion of a
vehicle's route up to time $t$ must visit every node visited by that vehicle by time $t$,
\eqref{onep3} ensures that this route indeed has length at most $t$, and finally
\eqref{path3} seeks to encode that the route forms a path. (Note that constraints
\eqref{path3} are clearly valid, and one could also include the constraints
$\sum_{a\in\dt^{\out}(r_i)}z^i_{a,t}\leq 1$ for all $i$, $t$.) 

Assuming $\Time$ is polynomially bounded, it is easy to design a separation oracle for the
exponentially-many cut constraints \eqref{jcov3}; hence, one can solve \eqref{lp3}
efficiently. We use this LP to obtain LP-relative guarantees for \kmlp (Section~\ref{kmlp})
which turn out to be quite versatile and extend easily to yield the same guarantees
for various generalizations. 

Intuitively, the difference between \eqref{lp1} and \eqref{lp3} boils down to the
following. Consider the tree version of \eqref{lp1}, $(\text{LP}_\Tcol)$. A solution to
this LP specifies for each time $t$ and vehicle $i$, a distribution over $r_i$-rooted
trees that together cover nodes to the extent dictated by the $x$ variables. Using
Theorem~\ref{arbpoly}, given a feasible solution $(x,z)$ to \eqref{lp3}, one can 
view $z$ as also specifying for each time $t$ and vehicle $i$ a distribution over
$r_i$-rooted trees covering nodes to the extents specified by the $x$ variables.
The difference however is that in the former case, 
{\em each tree in the support has length at most $t$}, whereas in the distribution
obtained from \eqref{lp3} one only knows that the {\em expected length} of a tree is at
most $t$.

\newcommand{\optbilp}{\ensuremath{\OPT_{\text{\ref{lp3}}}}}

\begin{theorem}[{\bf Relationship between \eqref{lp1}--\eqref{lp3} and the guarantees they
      yield}] \label{lprelate} 
We have the following. Recall that $\OPT_\Pc$ is the optimal value of \eqref{lp1}.
\begin{list}{(\roman{enumi})}{\usecounter{enumi} \itemsep=0ex \topsep=0.5ex
    \addtolength{\leftmargin}{0ex} \addtolength{\labelwidth}{\widthof{(iii)}}}
\item $\optbilp\leq\OPT_\Pc\leq\noptlp$ for every instance of multi-depot \kmlp.
  When $k=1$, we have $\OPT_\Pc=\noptlp$. 
\item For multi-depot \kmlp, we can efficiently compute a solution of cost at most
  $(8.4965+\ve)\OPT_\Pc$ for any $\ve>0$ (also Theorem~\ref{multithm}).
\item For single-depot \kmlp, we can efficiently compute a solution of cost at most
  $(2\mu^*+\ve)\optbilp$ for any $\ve>0$, where $\mu^*<3.5912$ (also Theorem~\ref{kmlpround}). 
\item When $k=1$, we can efficiently compute a solution of cost at most
  $(\mu^*+\ve)\optbilp$ for any $\ve>0$ (also Corollary~\ref{mlpround}). 
\end{list}
\end{theorem}

\begin{proof}
Parts (ii)--(iv) are simply restatements of the indicated theorems, whose proofs appear in
the corresponding sections. We focus on proving part (i).

For a $k$-tuple $\vP=(P_1,\ldots,P_k)$, we use $\vP(i)$ to denote $P_i$.
Let $(x,z)$ be a feasible solution to \eqref{lp2}. 
We may assume that constraints \eqref{jasgnconfig}, \eqref{jcovconfig} hold with equality
for all 
$v$ and $t$ since we can always shortcut paths past nodes without increasing their
length. We may also assume that if $z_{\vP,t}>0$ for $\vP=(P_1,\ldots,P_k)$, then the
any two $P_i$s are node-disjoint unless they originate from the same root node, in which
case this root is the only node they share.
We map $(x,z)$ to a feasible solution $(x',z')$ to \eqref{lp1} by setting
$z'^i_{P_i,t}=\sum_{\vP\in\vpcol_t:\vP(i)=P_i}z_{\vP,t}$ for all $i, t$ and 
$x'^i_{v,t}=\sum_{P\in\Pcol^i_t:v\in P}z'^i_{P,t}-\sum_{P\in\Pcol^i_{t-1}:v\in P}z'^i_{P,t-1}$ for all 
$i,v,t$, where we define $z'^i_{P,0}=0$ for all $P$ for notational convenience. 
It is easy to verify that $(x',z')$ is feasible for \eqref{lp1}. It's objective
value is 
\begin{eqnarray}
\sum_{v,t,i}tx'^i_{v,t} & = & \sum_{v,i}\sum_{t=1}^Tt\Bigl(\sum_{P\in\Pcol^i_t:v\in P}z'^i_{P,t}
-\sum_{P\in\Pcol^i_{t-1}:v\in P}z'^i_{P,t-1}\Bigr) \notag \\
& = & \sum_{v,i}\Bigl(\Time\sum_{P\in\Pcol^i_{\Time}:v\in P}z'^i_{P,\Time}
-\sum_{t=1}^{\Time-1}\sum_{P\in\Pcol^i_t:v\in P}z'^i_{P,t}\Bigr) 
\ = \ \sum_v\Bigl(\Time-\sum_{t=1}^{\Time-1}\sum_{\vP\in\vpcol_t:v\in\vP}z_{\vP,t}\Bigr)
\label{ineq1} \\
& = & \sum_v\Bigl(\Time-\sum_{t=1}^{\Time-1}\sum_{t'\leq t}x_{v,t'}\Bigr)
\ =\ \sum_v\Bigl(\Time-\sum_{t'=1}^{\Time}(\Time-t')x_{v,t'}\Bigr)
\ =\ \sum_{v,t'}t'x_{v,t'} \label{last}
\end{eqnarray}
The second equality in \eqref{ineq1} holds because
$\sum_i\sum_{P\in\Pcol^i_{t}:v\in P}z'^i_{P,t}
=\sum_i\sum_{\vP\in\vpcol_{t}:v\in\vP(i)}z_{\vP,t}
=\sum_{\vP\in\vpcol_t:v\in\vP}z_{\vP,t}$ since the paths comprising $\vP$ do not
share any non-root nodes; when $t=T$, this term is 1 since \eqref{jcovconfig} and
\eqref{jasgnconfig} hold at equality. 
The first and third equalities in \eqref{last} follow again from the fact that
\eqref{jcovconfig} and \eqref{jasgnconfig} hold at equality.
It follows that $\OPT_\Pc\leq\noptlp$.

Let $(x,z)$ be a feasible solution to \eqref{lp1}. It is easy to see that if we direct
each path in the support of $z$ away from its root and set
$z'^i_{a,t}=\sum_{P\in\Pcol^i_t:a\in P}z^i_{P,t}$, then $(x,z')$ is feasible for
\eqref{lp3}. Hence, $\optbilp\leq\OPT_\Pc$. 
\end{proof}

\section{\boldmath An LP-rounding 8.497-approximation algorithm for multi-depot \kmlp} \label{multi}  
We now prove the following theorem. Our approximation ratio of $8.497$ 
improves upon the previous-best
$12$-approximation~\cite{ChekuriK04,ChaudhuriGRT03} and matches the previous-best
approximation for single-depot \kmlp~\cite{FakcharoenpholHR07}.

\begin{theorem} \label{multithm}
For any $\ve>0$, we can compute a multi-depot-\kmlp solution of cost at most
$(8.4965+\ve)\cdot\OPT_\Pc$ in time $\poly\bigl(\text{input size},\frac{1}{\ve}\bigr)$. 
Thus, the integrality gap of \eqref{lp1} is at most $8.4965$.
\end{theorem}

Our algorithm is quite simple to describe. 
Let $(x,z)$ be the feasible solution to $\mltree{1+\e}$ 
returned by Lemma~\ref{lpsolve}, where we fix $\e$ later. 
We then choose time points that form a geometric sequence and do the following for
each time point $t$. For every $i=1,\ldots,k$, we sample a random tree from the
distribution $\{z^i_{Q,t}\}_{Q\in\Tcol^i_{(1+\e)t}}$, double and shortcut it to form a cycle and
traverse this cycle in a random direction to obtain a tour. For every $i$, we
concatenate the tours obtained for $i$ for each of the time points. 
We now describe the rounding procedure in detail and proceed to analyze it.

\newpage
{\small \hrule \vspace{-1pt}
\begin{algorithm} \label{multialg}
Given: a fractional solution $(x,z)$ of cost at most $\OPT_\Pc$ returned by
Lemma~\ref{lpsolve}.  
\end{algorithm}
\vspace{-10pt}
\begin{labellist}[M]
\item Let $\kp=1+\e$, and $1<c<e$ be a constant that we will fix later. Let $h=c^\Gm$ be a
random offset, where $\Gm$ is chosen uniformly at random from $[0,1)$.
For notational convenience, define $z^i_{Q,t}$ for all $t\geq 1$, $i$,
$Q\in\Tcol^i_{\kp t}$ as follows:
set $z^i_{Q,t}=z^i_{Q,\floor{t}}$ if $\floor{t}\leq\Time$ and 
$z^i_{Q,t}=z^i_{Q,\Time}$ otherwise.
Define $t_j=hc^j$ for all $j\geq 0$. 

\item Repeatedly do the following for $j=0,1,2,\ldots$ until 
every non-root node is covered (by some tour).
For every $i=1,\ldots,k$, choose independently a random tree $Q$ from the distribution 
$\bigl\{z^i_{Q,t_j}/\kp\bigr\}_{Q\in\Tcol^i_{\kp t_j}}$. 
Double and shortcut $Q$ to get a cycle, and traverse this cycle clockwise or
counterclockwise with probability $\frac{1}{2}$ to obtain a tour $Z_{i,j}$. 

\item For every $i=1,\ldots,k$, concatenate the tours $Z_{i,0},Z_{i,1},\ldots$ to obtain
the route for vehicle $i$.
\end{labellist}
\vspace{-2pt}\hrule
} 

\vspace{-1ex}
\paragraph{Analysis.}
The analysis hinges on showing that for every iteration $j$ of step M2, the
probability $p_{v,j}$ that a node $v$ is not covered by the end of iteration $j$ can be
bounded (roughly speaking) in terms of the total extent to which $v$ is not 
covered by $(x,z)$ by time $t_j$ (Lemma~\ref{probv}). 
Substituting this into the expression bounding the expected latency of $v$ in terms of the
$p_{v,j}$s (part (iii) of Claim~\ref{helper}),  
we obtain that by suitably choosing the constant $c$, 
the expected latency of $v$ is roughly $8.497\cdot\lat_v$, 
where $\lat_v:=\sum_{t,i}tx^i_{v,t}$ (Lemma~\ref{latv}).   
This proves that the algorithm returns a solution of cost roughly $8.497\cdot\OPT_\Pc$. 
However, it is not clear if the algorithm as stated above has polynomial running time. But   
since Lemma~\ref{probv} implies that $p_{v,j}$ decreases geometrically with $j$,
one can terminate step M2 after a polynomial number of iterations and 
cover the remaining uncovered nodes incurring latency at most $\Time$ for each such node. 
This increases the expected cost by at most $\ve\OPT_\Pc$ but ensures polynomial running
time; see Remark~\ref{truncalg}. 

Let $t_{-1}=0$, and define $\Dt_j:=t_j-t_{j-1}$ for all $j\geq 0$. 
For $q\geq 1$, define $\sg(q)$ to be the smallest $t_j$ that is at least $q$. 
Consider a non-root node $v$. 
We may assume that $\sum_{i,t}x^i_{v,t}=1$.
Define $p_{v,j}=1$ for all $j<0$.
Define $y^i_{v,t}:=\sum_{t'\leq t}x^i_{v,t'}$ 
and $o'_{v,j}:=1-\sum_{i}y^i_{v,t_{j}}$; define $o'_{v,j}=1$ for all $j<0$. 
Define $\lat'_v:=\sum_{j\geq 0}o'_{v,j-1}\Dt_j$. 
Let $L_v$ denote the random latency of node $v$ in the solution constructed. 
Note that the $t_j$s, and hence, $\sg(q)$, the $o'_{v,j}$s and $\lat'_v$ are random
variables depending only on the random offset $h$.
For a fixed offset $h$, we use $\E^h[.]$ to denote the expectation with respect to all
other random choices, while $\E[.]$ denotes the expectation with respect to all random
choices.  

\begin{claim} \label{helper}
For any node $v$, we have: 
(i) $\lat'_v=\sum_{t,i}\sg(t)x^i_{v,t}$;
(ii) $\E[\lat'_v]=\frac{c-1}{\ln c}\cdot\lat_v$; and \linebreak
(iii) $\E^h[L_v]\leq\frac{\kp(c+1)}{c-1}\cdot\sum_{j\geq 0}p_{v,j-1}\Dt_j$ for any fixed $h$. 
\end{claim}

\begin{proof} 
Part (i) follows from the same kind of algebraic manipulation as used in the proof of 
part (i) of Theorem~\ref{lp2facts}. We have 
\begin{equation*}
\begin{split}
\sum_{t,i}\sg(t)x^i_{v,t} & =\sum_{j\geq 0}t_j\Bigl(\sum_{t=t_{j-1}+1}^{t_j}\sum_ix^i_{v,t}\Bigr)
=\sum_{j\geq 0}\Bigl(\sum_{d=0}^j\Dt_d\Bigr)\Bigl(\sum_{t=t_{j-1}+1}^{t_j}\sum_ix^i_{v,t}\Bigr) \\
& =\sum_{d\geq 0}\Dt_d\Bigl(\sum_{j\geq d}\sum_{t=t_{j-1}+1}^{t_j}\sum_ix^i_{v,t}\Bigr)
=\sum_{d\geq 0}\Dt_do'_{v,d-1}.
\end{split}
\end{equation*}

Part (ii) follows from part (i) since we show that $\E[\sg(q)]=\frac{c-1}{\ln c}\cdot q$
for all $q\geq 1$. Suppose $q\in[c^j,c^{j+1})$ for some integer $j\geq 0$.
Then 
$$
\E[\sg(q)]=\int_0^{\log_c q-j}c^{y+j+1}dy+\int_{\log_c q-j}^1c^{y+j}dy
=\frac{1}{\ln c}\cdot\Bigl(c^{\log_c q+1}-c^{j+1}+c^{j+1}-c^{\log_c q}\Bigr)
=\frac{c-1}{\ln c}\cdot q.
$$

For part (iii), say that node $v$ is covered in iteration $j\geq 0$ if $j$ is the smallest
index such that $v\in\bigcup_i V(Z_{i,j})$.
By definition, the probability of this event is $p_{v,j-1}-p_{v,j}$, and in this case the
latency of $v$ is at most 
$\kp(2t_0+2t_1+\ldots+2t_{j-1}+t_j)\leq\frac{\kp(c+1)}{c-1}\cdot t_j$.
So $\E^h[L_v]\leq\frac{\kp(c+1)}{c-1}\sum_{j\geq 0}(p_{v,j-1}-p_{v,j})t_j
=\frac{\kp(c+1)}{c-1}\sum_{j\geq 0}p_{v,j-1}(t_j-t_{j-1})$.
\end{proof}

\begin{lemma} \label{probv}
$p_{v,j}\leq\bigl(1-e^{-1/\kp}\bigr)o'_{v,j}+e^{-1/\kp}p_{v,j-1}$ for all $j\geq -1$, and
all $v$. 
\end{lemma}

\begin{proof}
For $j=-1$, the inequality holds since $o'_{v,-1}=1=p_{v,-2}$. Suppose $j\geq 0$.
We have $p_{v,j}\leq p_{v,j-1}\prod_i\bigl(1-\frac{y^i_{v,t_j}}{\kp}\bigr)$
since the probability that $v$ is visited by the $i$-th tour in iteration $j$ is
$\sum_{Q\in\Tcol^i_{\kp t_j}:v\in Q}z^i_{Q,t_j}/\kp\geq y^i_{v,t_j}/\kp$.
We have that $\prod_{i=1}^k\bigl(1-\frac{y^i_{v,t_j}}{\kp}\bigr)$ is at most
$$
\biggl(1-\frac{\sum_i y^i_{v,t_j}}{\kp k}\biggr)^k
=\biggl(1-\frac{1-o'_{v,j}}{\kp k}\biggr)^k
\leq\Bigl(1-\tfrac{1}{\kp k}\Bigr)^k+\Bigl(1-\bigl(1-\tfrac{1}{\kp k}\bigr)^k\Bigr)o'_{v,j}
\leq e^{-1/\kp}+\bigl(1-e^{-1/\kp}\bigr)o'_{v,j}.
$$
The first inequality follows since the geometric mean is at most the arithmetic mean; the
second follows since $f(b)=\bigl(1-\frac{1-b}{\kp k}\bigr)^k$ is a convex function of $b$;
the final inequality follows since $o'_{v,j}\leq 1$ and 
$\bigl(1-\frac{1}{\kp k}\bigr)^k\leq e^{-1/\kp}$. Plugging the above bound into the
inequality for $p_{v,j}$ yields the lemma.
\end{proof}

\begin{lemma} \label{latv}
$\E[L_v]\leq\frac{\kp(c+1)(1-e^{-1/\kp})}{(\ln c)(1-ce^{-1/\kp})}\cdot\lat_v$ for all $v$.
\end{lemma}

\begin{proof}
Fix an offset $h$. 
We have $\frac{c-1}{\kp(c+1)}\cdot\E^h[L_v]\leq A:=\sum_{j\geq 0}p_{v,j-1}\Dt_j$ by
Claim~\ref{helper} (iii). 
Using Lemma~\ref{probv}, we obtain that 
$A\leq\sum_{j\geq 0}\bigl(1-e^{-1/\kp}\bigr)o'_{v,j-1}\Dt_j+e^{-1/\kp}\sum_{j\geq 0}p_{v,j-2}\Dt_j
=\bigl(1-e^{-1/\kp}\bigr)\lat'_v+ce^{-1/\kp}A$, 
where the equality follows since $\Dt_0+\Dt_1=c\Dt_0$, and $\Dt_j=c\Dt_{j-1}$ for all
$j\geq 2$, and so $\sum_{j\geq 0}p_{v,j-2}\Dt_j=cA$.
So $A\leq\frac{1-e^{-1/\kp}}{1-ce^{-1/\kp}}\lat'_v$.
Taking expectation with respect to the random offset $h$, and plugging in the bound for
$\E[\lat'_v]$ in part (ii) of Claim~\ref{helper} yields the lemma.
\end{proof}

Taking $c=1.616$, for any $\ve>0$, we can take $\e>0$ suitably small so that 
$\frac{\kp(c+1)(1-e^{-1/\kp})}{(\ln c)(1-ce^{-1/\kp})}\leq
\frac{(c+1)(1-e^{-1})}{(\ln c)(1-ce^{-1})}+\ve\leq 8.4965+\ve$. This completes the proof of
Theorem~\ref{multithm}. 

\begin{remark} \label{truncalg}
If we truncate step M2 to $N=D+\kp\ln\bigl(\frac{n\Time}{\ve}\bigr)$ iterations, 
where $t_D=\sg(\Time)$, then by Lemma~\ref{probv}, 
$p_{v,N}\leq e^{-(N-D)/\kp}\leq\frac{\ve}{n\Time}$ since $o'_{v,J}=0$.
Each remaining uncovered node can be covered incurring latency at most $\Time$ (since
$\Time$ is a certifiable upper bound).
This adds at most $\ve\leq\ve\OPT_\Pc$ to the expected cost of the solution,
but ensures polynomial running time.
\end{remark}

\section{\boldmath A $7.183$-approximation algorithm for (single-depot) \kmlp} 
\label{klat} \label{kmlp}
We now describe algorithms for \kmlp having approximation ratios essentially
$2\mu^*<7.183$. This guarantee can be obtained both by rounding the bidirected LP
\eqref{lp3} and via more combinatorial methods. The LP-rounding algorithm is slightly
easier to describe, and the analysis extends easily to the generalizations considered in
Section~\ref{extn}. But it is likely less efficient than the combinatorial algorithm, and
its guarantee is slightly weaker, $2\mu^*+\ve$, due to the fact that we need to solve the
time-indexed formulation \eqref{lp3} either by ensuring that $\Time$ is polynomially
bounded, or via the alternative method sketched in Section~\ref{extn}, both of which
result in a $(1+\ve)$-factor degradation in the approximation.  
We describe the LP-rounding algorithm first (Section~\ref{lpround}) and then the
combinatorial algorithm (Section~\ref{comb}).  

\subsection{The LP-rounding algorithm} \label{lpround}
We prove the following theorem. Recall that $D=(V,A)$ is the digraph obtained by
bidirecting $G$.

\begin{theorem} \label{kmlpround}
Any solution $(x,z)$ to \eqref{lp3} can be rounded to a \kmlp-solution losing a factor of
at most $2\mu^*<7.1824$. Thus, for any $\ve>0$, we can compute a \kmlp-solution of cost at
most $(2\mu^*+\ve)\optbilp$ in time $\poly\bigl(\text{input size},\ln(\frac{1}{\ve})\bigr)$. 
\end{theorem}

The rounding algorithm follows the familiar template of finding a
collection of tours with different node coverages and stitching them together using a
concatenation graph. Let $(x,z)$ be a feasible solution to \eqref{lp3}.
Let $x'_{v,t}=\sum_i x^i_{v,t}$ and $z'_{a,t}=\sum_i z^i_{a,t}$. We will in fact only work
with $(x',z')$. (This also implies that we obtain the same $2\mu^*$-guarantee with respect
to an even weaker bidirected LP where we aggregate the $k$ vehicles' routes and use a
single set of $x_{v,t}$ and $z_{a,t}$ variables for all $v,a,t$.)
For notational convenience, define $x'_{r,0}=1$, $x'_{r,t}=0$ for all $t>0$, and
$x'_{v,0}=0$ for all $v\neq r$. 
To give some intuition behind the proof of Theorem~\ref{kmlpround}, 
the following lemma will be useful. The proof involves simple algebraic manipulation, and    
is deferred to the end of this section.

\begin{lemma} \label{cglpbnd}
Suppose for every time $t=0,1,\ldots,\Time$, we have a random variable
$(N_t,Y_t)\in[1,n]\times\R$ such that $(N_0,Y_0)=(1,0)$ with probability 1, and
$\E[N_t]\geq\sum_{u\in V}\sum_{t'=0}^tx'_{u,t'}$, $\E[Y_t]\leq \al t$, for all 
$t\in[\Time]$. Let $f$ be the lower-envelope curve of 
$\bigcup_{t=0}^{\Time}(\text{ support of }(N_t,Y_t))$. 
Then, $\int_1^nf(x)dx\leq\al\sum_{u\in V,t\in[\Time]}tx'_{u,t}$.
\end{lemma}

Lemma~\ref{cglpbnd} coupled with Corollary~\ref{cgcor} imply that if 
one could efficiently compute for each time $t$, a random collection of $k$ trees rooted
at $r$ such that (a) their union covers in expectation at least 
$\sum_{u\in V}\sum_{t'=0}^tx'_{u,t'}=1+\sum_{v\neq r,t'\in[t]}x'_{v,t'}$ nodes, 
and (b) the expected maximum length of a tree in the collection is at most $t$, then we
would achieve a $\mu^*$-approximation by mimicking the proof of part (i) 
of Theorem~\ref{bnsbnd}. 
We do not quite know how to achieve this. 
However, as in the proof of Theorem~\ref{pcstroll}, applying Theorem~\ref{arbpoly} to (a
scaled version of) $(z'_{a,t})_{a\in A}$, we can efficiently find {\em one} random
$r$-rooted tree that in expectation has cost at most $kt$ and covers at least
$\sum_{u\in V}\sum_{t'=0}^tx'_{u,t'}$ nodes. 
This is the chief source of our improvement over the $8.497$-approximation
in~\cite{FakcharoenpholHR07}: we achieve the target coverage of the $k$ vehicles (which in
our case is determined by an LP) whereas \cite{FakcharoenpholHR07} sequentially find
separate tours for each vehicle, which succeeds in covering only a constant-fraction of
the target number of nodes (which in their case is determined by the integer optimal 
solution).  

The flip side is that we need to do slightly more work to convert the object computed 
into a (random) collection of $k$ low-cost tours containing $r$. 
To convert a rooted tree, we Eulerify it, break the resulting cycle into $k$ segments,
and attach each segment to $r$. Thus, for each time $t$, we obtain a random collection of
$k$ trees rooted at $r$ satisfying property (a) above, and a relaxed form of (b): the
expected maximum length of a tree in the collection is at most $2t$.
Thus, 
we obtain a solution of cost at most $2\mu^*$ times the cost of $(x,z)$.
We now describe the rounding algorithm in more detail and proceed to analyze it.

\vspace{2ex}
{\small \hrule 
\begin{algorithm} \label{kmlprounding}
The input is a feasible solution $(x,z)$ to \eqref{lp3}.
Let $x'_{v,t}=\sum_{i}x^i_{v,t},\ z'_{a,t}=\sum_{i}z^i_{a,t'}$ for all $v, a, t$.
\end{algorithm}
\vspace{-10pt}
\begin{labellist}[R]
\item Initialize $C\assign\{(1,0)\}$, $\Qc\assign\es$. Let $K$ be such $Kz'_{a,t}$ is an
integer for all $a, t$. 
For $t\in[\Time]$, define $S(t)=\{u\in V: \sum_{t'=0}^t x'_{u,t'}>0\}$. 
(Note that $r\in S(t)$ for all $t>0$.)

\item For all $t=1,\ldots,\Time$, do the following. 
Apply Theorem~\ref{arbpoly} on the digraph $D$ with edge weights $\{Kz'_{a,t}\}_{a\in A}$ and
integer $K$ (and root $r$) to obtain a weighted arborescence family
$(\gm_1,Q^t_1),\ldots,(\gm_q,Q^t_q)$. 
For each arboresence $Q^t_\ell$ in the family, which we view as a tree, 
add the point $\bigl(|V(Q^t_\ell)\cap S(t)|,\frac{2c(Q^t_\ell)}{k}+2t\bigr)$ to $C$, and
add the tree $Q^t_\ell$ to $\Qc$. 

\item For all $\ell=1,\ldots,n$, compute $s_\ell=f(\ell)$, where 
$f:[1,n]\mapsto\R_+$ is the lower-envelope curve of $C$.
We show in Lemma~\ref{lpsbound} that for every corner point $\bigl(\ell,f(\ell)\bigr)$ of
$f$, 
there is some tree $Q^*_\ell\in\Qc$ and some time $t^*_\ell$
such that $\ell=|V(Q^*_\ell)\cap S(t^*_\ell)|$, 
$f(\ell)=\frac{2c(Q^*_\ell)}{k}+2t^*_\ell$, and 
$\max_{v\in Q^*_\ell\cap S(t^*_\ell)}c_{rv}\leq t^*_\ell$. 

\item Find a shortest $1\leadsto n$ path $P_{C}$ in the concatenation graph
$\cg(s_1,\ldots,s_n)$. 

\item For every node $\ell>1$ on $P_{C}$, do the following. Double and shortcut $Q^*_\ell$
to obtain a cycle. Remove nodes on this cycle that are not in $S(t^*_\ell)$ by
shortcutting past such nodes.
Break this cycle into $k$ segments, each of length at most $2c(Q^*_\ell)/k$ and add edges
connecting the first and last vertex of each segment to $r$. This yields a collection of
$k$ cycles; traverse each resulting cycle in a random direction to obtain a collection of
$k$ tours $Z_{1,\ell}, \ldots, Z_{k,\ell}$.

\item For every $i=1,\ldots,k$, concatenate the tours $Z_{i,\ell}$ for nodes $\ell$ on
  $P_C$ to obtain vehicle $i$'s route.
\end{labellist}
\hrule
} 

\paragraph{Analysis.}
We first prove Lemma~\ref{cglpbnd}.
Lemma~\ref{lpsbound} utilizes this to bound $\int_1^n f(x)dx$, 
and shows that corner points of $f$ satisfy the properties stated in step R3. 
The latter allows us to argue that the solution returned has cost at most the length of
$P_C$ in the concatenation graph. Combining these facts with Corollary~\ref{cgcor} yields 
the $2\mu^*$ approximation ratio and completes the proof of Theorem~\ref{kmlpround}. 

\begin{proofof}{Lemma~\ref{cglpbnd}}
Note that $f$ is strictly increasing: for $x',x\in[1,n]$ with $x'<x$, since 
$x'=\frac{x'-1}{x-1}\cdot x+\frac{x-x'}{x-1}\cdot 1$ and $f(1)=0$, we have
$f(x')\leq\frac{x'-1}{x-1}f(x)<f(x)$. 
So we can write 
\begin{equation}
\frac{\int_1^nf(x)dx}{\al}=\int_1^n\Bigl(\int_0^{f(x)/\al}dy\Bigr)dx
=\int_0^{f(n)/\al}dy\bigl(\int_{f^{-1}(\al y)}^ndx\bigr)=\int_0^{f(n)/\al}\bigl(n-f^{-1}(\al y)\bigr)dy.
\label{integ}
\end{equation}
Note that $f(n)\leq\al\Time$. For any $t=0,1,\ldots,\Time$, the point
$(\E[N_t],\E[Y_t])=\E[(N_t,Y_t)]$ lies in the convex hull of the support of $(N_t,Y_t)$,
and so $f(\E[N_t])\leq\E[Y_t]\leq \al t$ and hence, 
$\E[N_t]\leq f^{-1}(\al t)$. So we can bound \eqref{integ} by
\begin{equation*}
\begin{split}
\sum_{t=1}^\Time\int_{t-1}^t\bigl(n-f^{-1}(\al y)\bigr)dy
& \leq\sum_{t=1}^\Time\bigl(n-f^{-1}(\al(t-1))\bigr) \\
& \leq\sum_{t=1}^\Time\Bigl(n-\sum_{u\in V}\sum_{t'=0}^{t-1}x'_{u,t'}\Bigr)
\leq\sum_{t=1}^\Time\sum_{u\in V,t'\geq t}x'_{u,t'}
=\sum_{u\in V}\sum_{t'=1}^\Time t'x'_{u,t'}. \qedhere
\end{split} 
\end{equation*}
\end{proofof}

\begin{lemma} \label{lpsbound}
(i) $\int_1^nf(x)dx\leq 4\sum_{u\in V,t\in[\Time]}tx'_{u,t}$. 
(ii) If $\bigl(\ell,f(\ell)\bigr)$ is a corner point of $f$, then there is a tree $Q^*_\ell$
and time $t^*_\ell$ satisfying the properties stated in step S3.
\end{lemma}

\begin{proof}
Consider any $t\in[\Time]$. The weighted arborescence family
$(\gm_1,Q^t_1),\ldots,(\gm_q,Q^t_q)$ yields a distribution over arborescences, where we
pick arborescence $Q^t_\ell$ with probability $\gm_\ell/K$. Let
$N_t$ and $Y_t$ be the random variables denoting $|V(Q^t_\ell)\cap S(t)|$ and
$\frac{2c(Q^t_\ell)}{k}+2t$ respectively. Define $\ld_D(r,r)=K$.
By Theorem~\ref{arbpoly}, we have 
$\E[c(Q^t_\ell)]\leq\sum_a c_az'_{a,t}\leq kt$, so $\E[Y_t]\leq 4t$; also,
$\E[N_t]\geq\sum_{u\in S(t)}\ld_D(r,u)/K\geq\sum_{u\in S(t)}\sum_{t'=0}^tx'_{u,t'}
=\sum_{u\in V}\sum_{t'=0}^tx'_{u,t'}$.
So the random variables $(N_t,Y_t)_{t\in[\Time]}$ and $(N_0,Y_0)=(1,0)$ and the curve $f$
satisfy the conditions of Lemma~\ref{cglpbnd} with $\al=4$, and its conclusion proves part
(i). 

For part (ii), corner points of $f$ are points of $C$. So if
$\bigl(\ell,f(\ell)\bigr)\in C$ is a corner point that was added to $C$ in iteration 
$t$ of step R2, then set $t^*_\ell=t$, and $Q^*_\ell$ to be the tree
added in this iteration. By definition, we have $|V(Q^*_\ell)\cap S(t^*_\ell)|=\ell$,
$f(\ell)=\frac{2c(Q^*_\ell)}{k}+2t^*_\ell$. Also
$\max_{v\in Q^*_\ell\cap S(t^*_\ell)}c_{rv}\leq t^*_\ell$, since $v$ can only lie in
$S(t^*_\ell)$ if $c_{rv}\leq t^*_\ell$, due to constraint \eqref{jasgn3}.
\end{proof}

\begin{proofof}{Theorem~\ref{kmlpround}}
We claim that the solution returned by Algorithm~\ref{kmlprounding} has cost at most the
length of $P_C$ in the concatenation graph $\cg(s_1,\ldots,s_n)$. 
Combining this with Corollary~\ref{cgthm} and part (i) of Lemma~\ref{sbound}, we obtain
that the total latency is at most 
$\frac{\mu^*}{2}\int_1^nf(x)dx\leq 2\mu^*\sum_{v,t}tx'_{v,t}=2\mu^*\sum_{v,i,t}tx^i_{v,t}$.

The proof of the claim is similar to the one in~\cite{GoemansK98} for single-vehicle \mlp.  
Consider an edge $(o,\ell)$ of $P_C$.
By Theorem~\ref{cgthm}, $\bigl(\ell,f(\ell)\bigr)$ is a corner point of $f$, so there
exist $Q^*_\ell$, $t^*_\ell$ satisfying the properties stated in step R3. It follows that
when we transform $Q^*_\ell$ into $k$ cycles containing only nodes of $S(t^*_\ell)$, each
cycle has length at most $f(\ell)=s_\ell$.

Suppose inductively that we have covered at least $o$ nodes by the partial solution
constructed by stitching tours corresponding to the nodes on $P_C$ up to and including
$o$. Consider the additional contribution to the total latency when we concatenate
tour $Z_{i,\ell}$ to vehicle $i$'s current route, for $i=1,\ldots,k$. 
The resulting partial solution covers at least $\ell$ nodes.
A node covered in this step incurs additional latency at most $\frac{s_\ell}{2}$ since we
traverse the cycle containing it (which has cost at most $s_\ell$) in a random direction. 
A node that is still uncovered after this step incurs additional latency at most $s_\ell$.
There are at most $\ell-o$ new nodes covered in this step, and at most $n-\ell$ uncovered
nodes after this step, so the increase in total latency is at most
$\frac{s_\ell}{2}(\ell-o)+s_\ell(n-\ell) = s_\ell\bigl(n-\frac{o+\ell}{2}\bigr)$,
which is exactly the length of $(o,\ell)$ edge in $\cg(s_1,\ldots,s_n)$. Therefore, by
induction the total latency is at most the length of $P_C$ in $\cg(s_1,\ldots,s_n)$. 
\end{proofof}

\begin{corollary} \label{mlpround}
For single-vehicle \mlp, any solution $(x,z)$ to \eqref{lp3} can be rounded losing a
factor of at most $\mu^*<3.5912$. Hence, for any $\ve>0$, we can compute an \mlp-solution
of cost at most $(\mu^*+\ve)\optbilp$ in time 
$\poly\bigl(\text{input size},\frac{1}{\ve}\bigr)$. 
\end{corollary}

\begin{proof}
This follows from essentially Algorithm~\ref{kmlprounding} and its analysis. The
improvement comes because we no longer need to break up a tree into $k$ tours. So in step
R2, for each tree $Q^t_\ell$ in the weighted arborescence family obtained for time $t$, we
add the point $\bigl(|V(Q^t_\ell)\cap S(t)|,2c(Q^t_\ell)\bigr)$ to $C$, and in part(i) of
Lemma~\ref{lpsbound}, we have the stronger bound 
$\int_1^nf(x)dx\leq 2\sum_{u\in V,t\in[\Time]}tx'_{u,t}$, which yields the $\mu^*$
approximation.  
\end{proof}

\subsection{The combinatorial approximation algorithm} \label{comb}
The combinatorial algorithm for \kmlp follows a similar approach as the LP-rounding
algorithm. The difference is that instead of using an LP to determine the target
coverage of the $k$ vehicles and maximum length of each vehicle's route, we now seek to
match the target coverage and length bound of an optimal $(k,\ell)$-bottleneck-stroll.   
Corollary~\ref{pccor} shows that one can efficiently find a rooted tree or bipoint tree  
that is at least as good (in terms of both total cost and node-coverage) as the optimal 
$(k,\ell)$-bottleneck-stroll solution for all $\ell$, and again this is where we score
over the algorithm in~\cite{FakcharoenpholHR07}. 
Again, we need to convert the object computed into a collection of $k$ tours, and
Theorem~\ref{cgthm} implies that a bipoint tree can be handled by handling the trees
comprising it. 
We convert a tree into $k$ tours as before, 
but to bound the cost of each resulting tour, 
we now need to ``guess'' the node furthest from $r$ covered by an optimal
$(k,\ell)$-bottleneck stroll solution, and apply Corollary~\ref{pccor} with more-distant
nodes removed; this ensures that each resulting tour has cost at most
$4\cdot\optbns{k,\ell}$.   
Hence, mimicking the proof of Theorem~\ref{bnsbnd} (i) shows that
we obtain a solution of cost at most $2\mu^*\cdot\bnslb<7.183\cdot\bnslb$. 
The algorithm and analysis are very similar to that in Section~\ref{lpround}.

\vspace{2ex}
{\small \hrule \vspace{-2pt}
\begin{algorithm} \label{alg_single_depot} \label{kmlpalg}
\end{algorithm}
\vspace{-10pt}
\begin{labellist}[S]
\item Initialize $C\assign\es$, $\Qc\assign\es$.
Let $v_1=r, v_2, \ldots, v_n$ be the nodes of $G$ in order of increasing distance 
from the root. Let $G_j=(V_j,E_j)$ the subgraph of $G$ induced by
$V_j:=\{v_1,\ldots,v_j\}$. 

\item For all $j,\ell=1,\ldots,n$, do the following. Use part (i) of Corollary~\ref{pccor}
with input graph $G_j$ and target $\ell$ to compute a rooted tree $Q_{j\ell}$ or rooted
bipoint tree $(a_{j\ell},Q_{j\ell}^1,b_{j\ell},Q_{j\ell}^2)$. In the former case, 
add the point $\bigl(|V(Q_{j\ell})|,\frac{2c(Q_{j\ell})}{k}+2c_{rv_j}\bigr)$
to $C$, and add the tree $Q_{j\ell}$ to $\Qc$. 
In the latter case, 
add the points 
$\bigl(|V(Q^1_{j\ell})|,\frac{2c(Q^1_{j\ell})}{k}+2c_{rv_j}\bigr)$,
$\bigl(|V(Q^2_{j\ell})|,\frac{2c(Q^2_{j\ell})}{k}+2c_{rv_j}\bigr)$ to $C$, and add the
trees $Q^1_{j\ell}$, $Q^2_{j\ell}$ to $\Qc$.

\item For all $\ell=1,\ldots,n$, compute 
$s_\ell=f(\ell)$, where $f:[1,n]\mapsto\R_+$ is the lower-envelope curve of $C$.
We show in Lemma~\ref{sbound} that for every corner point $\bigl(\ell,f(\ell)\bigr)$ of
$f$, there is some tree $Q^*_\ell\in\Qc$ and some index $j^*_\ell$ such that
$\ell=|V(Q^*_\ell)|$, 
$f(\ell)=\frac{2c(Q^*_\ell)}{k}+2c_{rv_{j^*_\ell}}$, and 
$\max_{v\in Q^*_\ell}c_{rv}\leq c_{rv_{j^*_\ell}}$. 

\item Find a shortest $1\leadsto n$ path $P_{C}$ in the concatenation graph
$\cg(s_1,\ldots,s_n)$. 

\item For every node $\ell>1$ on $P_{C}$, do the following. Double and shortcut $Q^*_\ell$
to obtain a cycle. 
Break this cycle into $k$ segments, each of length at most $2c(Q^*_\ell)/k$ and add edges
connecting the first and last vertex of each segment to $r$. This yields a collection of
$k$ cycles; traverse each resulting cycle in a random direction to obtain a collection of
$k$ tours $Z_{1,\ell}, \ldots, Z_{k,\ell}$.

\item For every $i=1,\ldots,k$, concatenate the tours $Z_{i,\ell}$ for nodes $\ell$ on
  $P_C$ to obtain vehicle $i$'s route.
\end{labellist}
\hrule
} 

\begin{lemma} \label{lem_bns_bound} \label{sbound}
(i) $s_\ell\leq 4\cdot\optbns{k,\ell}$ for all $\ell=1,\ldots,n$. \quad
(ii) If $\bigl(\ell,f(\ell)\bigr)$ is a corner point of $f$, then there is a tree $Q^*_\ell$
and index $j^*_\ell$ satisfying the properties stated in step S3.
\end{lemma}

\begin{proof}
For part (i), suppose that $v_j$ is the node furthest from $r$ that is covered by some
optimal $(k,\ell)$-bottleneck-stroll solution, so $c_{rv_j}\leq\optbns{k,\ell}$. 
Then, given part (i) of Corollary~\ref{pccor}, 
in iteration $(j,\ell)$ of step S2, 
we add one or two points to $C$ such that the point 
$\bigl(\ell,\frac{2z}{k}+2c_{rv_{j}}\bigr)$, for some $z\leq k\cdot\optbns{k,\ell}$, lies in
the convex hull of the points added. Therefore, $s_\ell=f(\ell)\leq 4\cdot\optbns{k,\ell}$
since $f$ is the lower-envelope curve of $C$.

The proof of part (ii) is essentially identical to the proof of Lemma~\ref{lpsbound} (ii).
\end{proof}

\begin{theorem} \label{kmlpthm}
Algorithm \ref{alg_single_depot} returns a solution of cost at most
$2\mu^*\cdot\bnslb$. Hence, it is a $2\mu^*$-approximation algorithm for \kmlp.
\end{theorem}

\begin{proof}
We claim that the solution returned has cost at most the length of $P_C$ in the
concatenation graph $\cg(s_1,\ldots,s_n)$. 
Combining this with Lemma~\ref{sbound} and Theorem~\ref{cgthm}, we obtain that the total
latency is at most 
$\frac{\mu^*}{2}\sum_{\ell=1}^n s_\ell\leq 2\mu^*\cdot\bnslb$, where $\mu^* < 3.5912$.  

Consider an edge $(o,\ell)$ of $P_C$.
By Theorem~\ref{cgthm}, $\bigl(\ell,f(\ell)\bigr)$ is a corner point of $f$, so there
exist $Q^*_\ell$, $j^*_\ell$ satisfying the properties stated in step S3. It follows that
when we transform $Q^*_\ell$ into $k$ cycles, each cycle has length at most
$f(\ell)=s_\ell$. Given this, the rest of the proof proceeds identically as that of
Theorem~\ref{kmlpround}. 
\end{proof}

\section{Extensions} \label{extn}

We now consider some extensions of multi-depot \kmlp and showcase the versatility of our
algorithms by showing that our guarantees extend mostly with little effort to 
these problems. 

\begin{theorem} \label{extn-thm}
For any $\ve>0$, we can compute a $(\rho+\ve)$-approximation for the following
generalizations of multi-depot \kmlp in time $\poly\bigl(\text{input size},\frac{1}{\ve}\bigr)$: 
(i) weighted sum of node latencies: $\rho=8.4965$; (ii) node-depot service constraints:
$\rho=8.4965$; and (iii) node service times: $\rho=8.9965$. 
The approximation ratios for (i) and (iii) improve to $(7.1824+\ve)$ for the single-depot
version.  
\end{theorem}

In some of the settings below, we will only be able to ensure that our certifiable upper
bound $\Time$ on the maximum latency of a node is such that $\log\Time$ (as opposed to
$\Time$) is polynomially bounded. 
This means that the resulting extension of \eqref{lp1} may have exponentially many
variables {\em and} constraints. 
To circumvent this difficulty, we sketch below an approach 
for efficiently computing a $(1+\e)$-approximate solution to \eqref{lp1} that only relies
on $\log\Time$ being polynomially bounded, with a $(1+\e)$-violation in some
constraints in the same sense as in Lemma~\ref{lpsolve}: namely, for each $i$ and any
time point $t$ under consideration, we use $r_i$-rooted {\em trees} of length $(1+\e)t$
and total fractional weight at most $(1+\e)$ (instead of a collection of $r_i$-rooted
paths of length $t$ of total fractional weight at most $1$) to cover nodes to the extent
they are covered by time $t$. 
We call such a solution a {\em multicriteria $(1+\e)$-approximate solution}. 
This approach easily extends to solve
the various LPs encountered below. 

\vspace{-1ex}
\paragraph{\boldmath Solving \eqref{lp1} when $\log\Time$ is polynomially bounded.} 
Borrowing an idea from~\cite{ChakrabartyS11}, we move to a compact version of \eqref{lp1}
where 
we only have variables $\{x^i_{v,t}\}$, $\{z^i_{P,t}\}$, and constraints \eqref{onep},
\eqref{jcov} for $t$s in a polynomially-bounded set $\TS$. 
We set $\TS:=\{\Time_0,\ldots,\Time_D\}$, where $\Time_j=\ceil{(1+\e)^j}$, and
$D=O(\log\Time)=\poly(\text{input size})$ is the smallest integer such that
$\Time_D\geq\Time$. We use $\lppathts$ to denote this LP. The ``tree-version'' of
$\lppathts$ is obtained similarly from $(\text{LP}_\Tcol)$ and denoted $\lptreets$.

Define $\Time_{-1}=0$
Given a solution $(x,z)$ to \eqref{lp1}, where $t$ ranges from $1$ to $\Time$, we can
define $(x',z')$ as follows: set $z'^i_{P,t}=z^i_{P,t}$ for all $i, P\in\Pcol_t$ and
$t\in\TS$; set $x'^i_{v,T_j}=\sum_{t\in T_{j-1}+1}^{T_j}x^i_{v,t}$ for all $i, v$, and
$T_j\in\TS$. 
It is not hard to see that $(x',z')$ is feasible to $\lppathts$ and that its
cost is at most $(1+\e)$ times the cost of $(x,z)$. Thus, the optimal value of $\lppathts$
is at most $(1+\e)\OPT_\Pc$. 
Conversely, given a solution $(x',z')$ to $\lppathts$, setting $x^i_{v,t}$ equal to
$x'^i_{v,t}$ if $t\in\TS$ and 0 otherwise, and $z^i_{P,t}=z'^i_{P,\Time_j}$ for all
$t\in[\Time_j,\Time_{j+1})$ and all $j$, yields a feasible solution to \eqref{lp1} of the
same cost. 

Since $\lppathts$ is an LP of the same form as \eqref{lp1} but
with polynomially many variables, we can approximately solve it in the sense of
Lemma~\ref{lpsolve}: for any $\e>0$, we can obtain in time 
$\poly\bigl(\text{input size},\frac{1}{\e}\bigr)$ a solution to $\lptreets$ of cost at
most $\OPT_{\lppathts}\leq(1+\e)\OPT_\Pc$ with a $(1+\e)$-violation in some 
constraints. This in turn yields a solution to \eqref{lp1} of no greater cost and with the 
same $(1+\e)$-violation in some constraints.

Observe that the above idea of restricting time points to the polynomially-bounded set
$\{\Time_0,\ldots,\Time_D\}$ also applies to \eqref{lp3} and shows that we can obtain a
feasible solution to \eqref{lp3} of cost at most $(1+\e)\optbilp$ in time
$\poly\bigl(\text{input size},\frac{1}{\e}\bigr)$ while only assuming that
$\log\Time=\poly(\text{input size})$.

\subsection{Weighted sum of node latencies} 
Here, we have nonnegative node weights
$\{w_v\}$ and want to minimize the weighted sum $\sum_vw_v(\text{latency of }v)$ of node
latencies. We again have the upper bound $\Time=2n\lb$ on the maximum latency of a node.
We cannot use scaling and rounding to ensure that $\Time=\poly(\text{input size})$, but
note that $\log\Time=\poly(\text{input size})$.

For multi-depot \kmlp, we consider \eqref{lp1} with the objective modified to take into
account the node weights. We can obtain a multicriteria $(1+\e)$-approximate solution to
the resulting LP as described above. We round this as before; this works since
Lemma~\ref{latv} remains unchanged and bounds the expected latency of each node in terms
of the latency it incurs under the LP solution. The only minor change is that in the
truncated version (Remark~\ref{truncalg}), we set 
$N=D+\kp\ln\bigl(\frac{(\max_v w_v)n\Time}{\ve}\bigr)$  
since covering an uncovered node at the end incurs weighted latency at most 
$(\max_v w_v)\Time$.  

\paragraph{\boldmath A $7.183$-approximation for \kmlp.} 
Both the LP-rounding and the combinatorial algorithms in Section~\ref{kmlp} can be
extended to this setting. We describe the LP-rounding algorithm here; the extension of the 
combinatorial algorithm is descibed in Appendix~\ref{combwtextn}.
We consider \eqref{lp3} with the weighted-latency objective and obtain a
$(1+\e)$-approximate solution $(x,z)$ to this LP. We round this losing a  
$2\mu^*$-factor in a very similar fashion to Algorithm~\ref{kmlprounding} in
Section~\ref{lpround}. 
We may assume via scaling that all weights are integers, and $w_r=1$. 
Let $W=\sum_{u\in V} w_u$. 
A naive extension of the algorithm in Section~\ref{kmlp} would be to create $w_v$ nodes
co-located at $v$ and include a node in the concatenation graph for every possible weight
value from $0$ to $\sum_v w_v$. But this only yields pseudopolynomial running time. 
Instead, we proceed as follows.

Let $\TS:=\{\Time_0,\ldots,\Time_D\}$ be the time points that we consider when solving
\eqref{lp3} approximately, where $\Time_j=\ceil{(1+\e)^j}$ for all $j\geq 0$, and 
$D=O(\log\Time)=\poly(\text{input size})$ is the smallest integer such 
that $\Time_D\geq\Time$. 
In Algorithm~\ref{kmlprounding}, we always only consider time points in $\TS$. 
Step R1 is unchanged. 
In step R2, for each time $t\in\TS$ and each arborescence $Q^t_\ell$ of the weighted
arborescence family obtained for time $t$
we now add the point $\bigl(w(V(Q^t_\ell)\cap S(t)),\frac{2c(Q^t_\ell)}{k}+2t\bigr)$ to
$C$, and as before, add $Q^t_\ell$ to $\Qc$.

Let $f:[1,W]\mapsto\R_+$ be the lower-envelope curve of $C$. 
We claim that the shortest path $P_C$ in the concatenation graph
$\cg\bigl(f(1),\ldots,f(W)\bigr)$ can be computed efficiently. 
This is because by Theorem~\ref{cgthm}, the shortest
path only uses nodes corresponding to corner points of $f$. So the shortest path remains
unchanged if we only consider edges in the concatenation graph incident to such
nodes. This subgraph of the concatenation graph has polynomial size (and can be computed)
since all corner points of $f$ must be in $C$ and $|C|=O(D)$. Moreover, as in part (ii) of
Lemma~\ref{lpsbound}, every corner point $\bigl(\ell,f(\ell)\bigr)$ of $f$ corresponds to
some tree $Q^*_\ell\in\Qc$ and some $t^*\ell\in\TS$ such that 
$\ell=w\bigl(V(Q^*_\ell)\cap S(t^*_\ell)\bigr)$,
$f(\ell)=\frac{2c(Q^*_\ell)}{k}+2t^*_\ell$, and 
$\max_{v\in Q^*_\ell\cap S(t^*_\ell)}c_{rv}\leq t^*_\ell$. Given this, steps R5, R6 are
unchanged. 

The analysis also proceeds as before. Mimicking the proof of Theorem~\ref{kmlpround}, we
can again argue that the solution returned has cost at most the length of $P_C$ in
$\cg\bigl(f(1),\ldots,f(W)\bigr)$. To complete the analysis, utilizing
Corollary~\ref{cgcor}, we need to bound $\int_1^Wf(x)dx$. 
Recall that $x'_{v,t}=\sum_i x_{v,t}$ for all $v\neq r, t\in\TS$. As before, 
define $x'_{r,0}=1$, $x'_{r,t}=0$ for all $t>0$. 
Also, define $x'_{v,t}=0$ for all $v\neq r$ and all $t<\Time_D, t\notin\TS$.  
Generalizing part (i) of Lemma~\ref{lpsbound}, we show that $\int_1^Wf(x)dx$ is
at most $4\sum_{u\in V,t\in\TS}w_u\cdot tx'_{u,t}$. 

Define $\Time_{-1}:=0$.
Dovetailing the proof of Lemma~\ref{cglpbnd}, we have that 
$\bigl(\int_1^Wf(x)dx\bigr)/4=\int_0^{f(W)/4}\bigl(W-f^{-1}(4y)\bigr)dy$.
Note that $f(W)\leq 4\Time_D$.
For any $t=\Time_j\in\TS$, we include all points generated by arborescences in the
weighted arborescence family for $t$ in $C$. 
So we ensure that some point $(a,b)$, where $a\geq\sum_{u\in V}\sum_{t'=0}^tw_ux'_{u,t'}$, 
$b\leq 4t$ lies in the convex hull of $C$. So 
$f^{-1}(4t)\geq\sum_{u\in V}\sum_{t'=0}^tw_ux'_{u,t'}$; this also holds for $t=0$. 
So as in the proof of Lemma~\ref{cglpbnd}, we have
\begin{equation*}
\begin{split}
\int_0^{f(W)/4}\bigl(W-f^{-1}(4y)\bigr)dy
& \leq\sum_{j=0}^D(\Time_j-\Time_{j-1})\bigl(W-f^{-1}(4\Time_{j-1})\bigr)
=\sum_{j=0}^D(\Time_j-\Time_{j-1})\bigl(W-\sum_{u\in V}\sum_{t'=0}^{\Time_{j-1}}w_ux'_{u,t'}\bigr) \\
& =\sum_{j=0}^D(\Time_j-\Time_{j-1})\sum_{u\in V}w_u\sum_{t'\geq\Time_j}x'_{u,t'}
=\sum_{u\in V}\sum_{t'\in\TS}w_u\cdot t'x'_{u,t'}
\end{split}
\end{equation*}

\subsection{Node-depot service constraints}
In this setting, we are given a set $S_v\sse R$ of depots for each node $v$, and 
$v$ must be served by a vehicle originating at a depot in $S_v$. 
The $8.497$-approximation algorithm extends in a straightforward manner.
We now define $\lb:=\max_v\min_{i\in S_v}c_{r_iv}$, and can again ensure that $\lb$, and
hence $\Time=2n\lb$ is polynomially bounded. We modify constraint \eqref{jasgn} of
\eqref{lp1} to $\sum_{t,i\in S_v}x^i_{v,t}\geq 1$, obtain a solution to the resulting LP via
Lemma~\ref{lpsolve}, and round it as before.

\subsection{Node service times} \label{ndsrv}
Here, each non-root node $v$ has a service time $d_v$ that is added to the latency of node
$v$, and every node visited after $v$ on the path of the vehicle serving $v$. 
Set $d_r=0$ for $r\in R$ for notational convenience.
We can set $\Time=\sum_v d_v+2n\lb$ as an upper bound on the maximum latency of a node. 

Let $c''_{uv}=c_{uv}+\frac{d_u+d_v}{2}$ for all $u,v$. Observe that the $c''_e$s form a
metric. 
We obtain a multicriteria $(1+\e)$-approximate solution $(x,z)$ to \eqref{lp1} with 
the $c''$-metric. Note that this LP is a valid relaxation since if $P$ is the
portion of a vehicle's route up to and including node $v$ then $c''(P)$ is at most the
latency of $v$.
We round $(x,z)$ as in Algorithm~\ref{multialg}. 
The additive 0.5 increase in
approximation comes from the fact that when we convert a tree $Q$ of $c''$-cost $t$ 
into a cycle $Z$, the expected contribution to the latency of a node $v\in Z$ 
is now at most $d_v+\frac{1}{2}\bigl(2c''(Q)-d_v\bigr)\leq t+\frac{d_v}{2}$.
Thus, we obtain an $8.997$-approximation.

\paragraph{\boldmath A $7.183$-approximation for \kmlp.} 
Define the {\em mixed length} of a path or tree $Q$ to be $c(Q)+d(V(Q))$. 
Defining the {\em directed} metric $c'_{u,v}=c_{uv}+d_v$ for all $u, v$, note that if we
have a rooted tree and we direct its edges away from $r$, then its $c'$-cost 
is exactly its mixed length (since $d_r=0$).
Again, both the LP-rounding and combinatorial algorithms in Section~\ref{kmlp} extend with
small changes. Essentially, the change is that we work with the $c'$-metric, which works
out in the LP-rounding algorithm since Theorem~\ref{arbpoly} does not depend in any way on
the edge costs, and works out in the combinatorial algorithm since Theorems~\ref{pcstroll}
and Corollary~\ref{pccor} also apply with the $c'$-metric and yield analogous statements
where the $c$-cost is replaced by the mixed-length objective. The only thing to verify is
that the procedure for converting a tree $Q$ of mixed length (i.e., $c'$-cost) at most
$kt$ into $k$ tours ensures that the expected contribution to the latency of a node 
$v\in Q$ is at most $t$. 
We describe the LP-rounding algorithm here and the combinatorial algorithm in
Appendix~\ref{combndwtextn}. 

Let $(x,z)$ be a $(1+\e)$-approximate solution $(x,z)$ to \eqref{lp3} with arc-costs
$\{c'_a\}_{a\in A}$ (instead of the $c$-metric), obtained by considering time points in
$\TS:=\{\Time_0,\ldots,\Time_D\}$. Here $\Time_j=\ceil{(1+\e)^j}$ for all $j\geq 0$, and 
$D=O(\log\Time)=\poly(\text{input size})$ is the smallest integer such 
that $\Time_D\geq\Time$. 
As before, let $x'_{v,t}=\sum_i x^i_{v,t}$ and $z'_{a,t}=\sum_i z^i_{a,t}$. 
Define $x'_{r,0}=1$, $x'_{r,t}=0$ for all $t>0$, and $x'_{v,0}=0$ for all $v\neq r$. 
In Algorithm~\ref{kmlprounding}, we always only consider time points in $\TS$. 
Steps R1, R4 are unchanged. The only change in steps R2, R3 is that the $c$-cost is
replaced by the mixed-length objective (i.e., the $c'$-cost of the out-tree rooted at
$r$). 

The main change is in step R5, where we need to be more careful in obtaining the
collection of $k$ tours $Z_{1,\ell},\ldots,Z_{k,\ell}$ that cover 
$V(Q^*_\ell)\cap S(t^*_\ell)$ from the rooted tree $Q^*_\ell$. 
We show that we can obtain these $k$ tours so that we have
$$
c(Z_{i,\ell})+2d(V(Z_{i,\ell}))\leq
2\cdot\frac{c(Q^*_\ell)+d\bigl(V(Q^*_\ell)\bigr)}{k}+2t^*_\ell
\qquad \frall i=1,\ldots,k.
$$ 
This follows from Lemma~\ref{break}, where we prove that we can obtain $k$ cycles
satisfying the above inequality; traversing each cycle in a random direction, clockwise or
counterclockwise, yields the desired $k$ tours.

\begin{lemma} \label{break}
Let $Q$ be a rooted tree. Let $S\sse V(Q)$ and $L=\max_{u\in S}(c_{ru}+d_u)$. 
We can obtain $k$ cycles $Z_1,\ldots,Z_k$ that together cover $S$ such that  
$c(Z_{i})+2d(V(Z_{i}))\leq 2\frac{c(Q)+d(V(Q))}{k}+2L$ for all $i=1,\ldots,k$.
\end{lemma}

\begin{proof}
Assume that $S$ contains a node other than the root $r$, otherwise, we can
take $Z_1,\ldots,Z_k$ to be the trivial ones consisting of only $r$.
Let $M=\frac{c(Q)+d(V(Q))}{k}$.
Pick an arbitrary node $w\in S$, $w\neq r$.
First, we double and shortcut $Q$ to remove nodes not in $S$ and repeat occurrences of
nodes, to obtain an $r$-$w$ path $P$ satisfying $c(P)\leq 2c(Q)$, 
and therefore, $c(P)+2d(v(P))\leq 2Mk$. Note that $c_{rv}+d_v\leq L$ for every $v\in P$.  
We obtain the cycles by snipping $P$ at appropriate places and joining the resulting
segments to the root. To bound the cost of each resulting cycle and argue that the
snipping process creates at most $k$ segments, we define two charges $\charge^1$ and
$\charge^2$ for each segment that, roughly speaking, sandwich the quantity of interest
$c(.)+2d(.)$ for each segment.
For $u, v\in P$, let $P_{uv}$ denote the portion of $P$ between (and including) nodes $u$ 
and $v$.  

Set $u_1=r$. We repeatedly do the following. 
Let $v$ be the first node after $u_i$ on $P$ 
such that one of the following holds. If there is no such node, then we terminate the
loop, and set $v_i=w$.

\begin{list}{(\roman{enumi})}{\usecounter{enumi} \topsep=0.5ex \itemsep=0ex
    \addtolength{\leftmargin}{-2ex} \addtolength{\labelwidth}{\widthof{(ii)}}} 
\item $c(P_{u_iv})+2d(V(P_{u_iv}))-d_{u_i}-d_{v}>2M$.
Note that in this case $P_{u_iv}$ consists of at least two edges, since otherwise we
have 
$$
c(P_{u_iv})+2d(V(P_{u_iv}))-d_{u_i}-d_{v}=c_{u_iv}+d_{u_i}+d_{v}\leq
c_{ru_i}+d_{u_i}+c_{rv}+d_{v}\leq 2L.
$$
We set set $u_{i+1}=v$, and $v_i$ to be the node immediately before $v$ on $Z$. Define   
$\charge^1_i=c(P_{u_iv_i})+2d(V(P_{u_iv_i}))-d_{u_i}-d_{v_i}$ and
$\charge^2_i=c(P_{u_iv})+2d(V(P_{u_iv}))-d_{u_i}-d_{v}$. We say that $P_{u_iv_i}$ is of
type (i).

\item $c(P_{u_iv})+2d(V(P_{u_iv}))-d_{u_i}-d_{v}\leq 2M<c(P_{u_iv})+2d(V(P_{u_iv}))-d_{u_i}$. 
We set $v_i=v$. Define
$\charge^1_i=c(P_{u_iv})+2d(V(P_{u_iv}))-d_{u_i}-d_{v}$ and 
$\charge^2_i=c(P_{u_iv})+2d(V(P_{u_iv}))-d_{u_i}$.
We say that $P_{u_iv_i}$ is of type (ii).
If $v=w$, then we terminate the loop; otherwise, we set $u_{i+1}$ to be the node
immediately after $v$ on $P$.  
\end{list}
We increment $i$ and repeat the above process.

\medskip
Suppose we create $q$ segments in the above process, i.e., $q$ is the value of the counter 
$i$ at termination. We first argue that $q\leq k$. To see this, note that $\charge^2_i$ is
certainly well defined for all $i=1,\ldots,q-1$, and by definition, $\charge^2_i>2M$ for
all $i=1,\ldots,q-1$. So $\sum_{i=1}^{q-1}\charge^2_i>2M(q-1)$. 
We claim that $\sum_{i=1}^{q-1}\charge^2_i\leq c(P)+2d(V(P))\leq 2Mk$, and therefore
$q-1<k$. To see the claim, notice
that every edge of $P$ contributes to at most one $\charge^2_i$. Also, every node $v$
contributes in total at most $2d_v$ to $\sum_{i=1}^{q-1}\charge^2_i$. This is certainly
true if $v\notin\{u_1,\ldots,u_q\}$; otherwise if $v=u_i$, then it contributes
$d_v$ to $\charge^2_i$, and possibly $d_v$ to $\charge^2_{i-1}$, if $P_{u_{i-1}v_{i-1}}$ is
of type (i).

Each segment $P_{u_iv_i}$ yields a cycle $Z_i$ by joining $u_i$ and $v_i$ to $r$; the
remaining $k-q$ cycles are the trivial ones consisting of only $\{r\}$. 
If $\charge^1_q$ has not been defined (which could happen if no node $v$ satisfies
(i) or (ii) when $i=q$), define it to be $c(P_{u_qv_q})+2d(V(P_{u_qv_q}))-d_{u_q}-d_{v_q}$. 
For $i=1,\ldots,q$, note that by definition, we have 
$c(P_{u_iv_i})+2d(V(P_{u_iv_i}))-d_{u_i}-d_{v_i}=\charge^1_i\leq 2M$.
So $c(Z_i)+2d(V(Z_i))=(c_{ru_i}+d_{u_i})+(c_{rv_i}+d_{v_i})+\charge^1_i\leq 2M+2t$.
\end{proof}

The remainder of the analysis dovetails the one in Section~\ref{lpround}.
Analogous to Lemma~\ref{lpsbound} (and as in the weighted-latency setting), we have that
the lower-envelope curve $f$ satisfies:
(i) $\int_1^nf(x)dx\leq 4\sum_{u\in V,t\in\TS}tx'_{u,t}$, and
and (ii) every corner point $\bigl(\ell,f(\ell)\bigr)$ satisfies the properties stated in
the modified step R3.
Finally, we argue that the cost of the solution returned is at most the length of the
shortest path $P_C$ in $\cg\bigl(f(1),\ldots,f(n)\bigr)$, which yields an approximation
guarantee of $2\mu^*(1+\e)$ (where $\mu^*<3.5912$).

As before, consider an edge $(o,\ell)$ of $P_C$. Assume inductively that we have covered
at least $o$ nodes by the partial solution obtained by concatenating tours corresponding
to the portion of $P_c$ up to and including $o$. 
By Lemma~\ref{break}, and since $\bigl(\ell,f(\ell)\bigr)$ is a corner point of $f$, each
$Z_{i,\ell}$ has mixed length at most $f(\ell)$. Also, concatenating $Z_{i,\ell}$
to vehicle $i$'s route, for all $i=1,\ldots,k$, we end up covering at least $\ell$ nodes. 
The increase in latency due to this step for a node that remains uncovered after this step
is at most $f(\ell)$. Consider a node $v$ that is covered in this step, and say 
$v\in Z_{i,\ell}$. Since $Z_{i,\ell}$ is obtained by traversing the corresponding cycle in
a random direction, the increase in latency of $v$ is at most
$$
d_v+\frac{1}{2}\cdot\Bigl(c(Z_{i,\ell})+d\bigl(V(Z_{i,\ell})\bigr)-d_v\Bigr)
\leq\frac{1}{2}\cdot\Bigl(c(Z_{i,\ell})+d\bigl(V(Z_{i,\ell})\bigr)+d_v\Bigr)
\leq\frac{c(Z_{i,\ell})+2d\bigl(V(Z_{i,\ell})\bigr)}{2}\leq\frac{f(\ell)}{2}.
$$
The last inequality follows from Lemma~\ref{break}.
Therefore, as before, the increase in total latency due to this step is at most
the length of the $(o,\ell)$ edge in $\cg\bigl(f(1),\ldots,f(n)\bigr)$. So by induction,
the total latency is at most the length of $P_c$ in $\cg\bigl(f(1),\ldots,f(n)\bigr)$.

\section{Proof of Theorem~\ref{lp2facts}} \label{altlp}
Part (iii) is simply a restatement of Lemma~\ref{lpsolve}, so we focus on parts (i) and (ii).

\medskip \noindent
{\em Proof of part (i).\ } The proof follows from some simple algebraic manipulations.
We express both the objective value of \eqref{lp2} and $\bnslb$ equivalently as the
sum over all time units $t$ of the number of uncovered nodes at time $t$ (i.e., after time
$t-1$), where for the LP, by ``number'' we mean the total extent to which nodes are not
covered. We then observe that this ``number'' for an LP solution is at least the
corresponding value in the expression for $\bnslb$.

Let $b^*_\ell=\optbns{k,\ell}$ for $\ell=1,\ldots,n$.
Note that $b^*_\ell$ is an integer for all $\ell$ since all $c_e$s are integers, and
$b^*_\ell=0$ for all $\ell\leq|R|$.
Let $(x,z)$ be an optimal solution to \eqref{lp2}. 
For convenience, we set $x_{v,t}=0=z_{\vP,t}$ for all $t>\Time$ and all
$v,\vP\in\vpcol_t$. Also set $x_{r_i,t}=0$ for all $t\geq 1$. 
Let $u$ index nodes in $V$.
Define $n^*_t=\max\{\ell: b^*_\ell\leq t\}$ for all $t\geq 1$, and 
$N_t=\sum_{u,t'\leq t}x_{u,t'}$ for all $t\geq 1$. 
Note that 
$N_t\leq\sum_u\sum_{\vP\in\vpcol_t: u\in\vP}z_{\vP,t}
=\sum_{\vP\in\vpcol_t}z_{\vP,t}|\{u: u\in\vP\}|\leq n^*_t$ for all $t$.

We now express both the objective value of \eqref{lp2} and $\bnslb$ equivalently as the
sum over all time units $t$ of the number of uncovered nodes at time $t$ (i.e., after time
$t-1$). This coupled with the fact that $N_t\leq n^*_t$ for all $t$ completes the
proof. We have 
$$
\bnslb=\sum_{\ell=1}^n b^*_\ell
=\sum_{\ell=1}^n\sum_{t=1}^{b^*_\ell}1=\sum_{t=1}^{b^*_n}\sum_{\ell: b^*_\ell\geq t}1
=\sum_{t=1}^{b^*_n}(n-n^*_{t-1})=\sum_{t\geq 1}(n-n^*_{t-1}).
$$
We also have 
\begin{equation*}
\begin{split}
\noptlp &= \sum_{t\geq 1,u}tx_{u,t}=\sum_t\bigl(\sum_{t'=1}^t 1\bigr)\sum_u x_{u,t} \\
& =\sum_{t'\geq 1}\sum_{t\geq t',u}x_{u,t}\geq\sum_{t'\geq 1}(n-N_{t'-1})\geq
\sum_{t'\geq 1}(n-n^*_{t'-1})\geq\bnslb. 
\end{split}
\end{equation*}

\medskip \noindent
{\em Proof of part (ii).\ } We prove the second statement, which immediately implies the
first. 
Let $\kp=1+\e$.
Let $(x,z)$ be a solution to $\configtree{\kp}$. 
The rounding procedure and its analysis are very similar to the one in Section~\ref{multi}.
Let $h=c^\Gm$, where $\Gm\sim U[0,1)$. At each time
$t_j:=hc^j$, for $j=0,1,\ldots$, we sample a tree configuration $\vQ=(Q_1,\ldots,Q_k)$
from the distribution $\bigl\{z_{\vQ,t_j}/\kp\bigr\}_{\vQ\in\vtcol_{\kp t}}$; we
convert each $Q_i$ into a cycle and traverse this cycle in a random direction to obtain a
tour $Z_{i,j}$. We then concatenate the tours $Z_{i,0},Z_{i,1},\ldots$ for all
$i=1\ldots,k$. 

Define $t_{-1}$, $\Dt_j$, $\lat_v$, $p_{v,j}$, $L_v$ for all $v, j$ as in the analysis in
Section~\ref{multi}. Define $y_{v,t}:=\sum_{t'\leq t}x_{v,t'}$ and $o'_{v,j}:=1-y_{v,t_j}$,
for all $v, t, j$; let $o'_{v,j}=1$ for all $v$ and $j<0$.
Let $\lat'_v:=\sum_{j\geq 0}o'_{v,j-1}\Dt_j$. Let $\E^h[.]$ and $\E[.]$ denote the
same quantities as in the analysis in Section~\ref{multi}.
Parts (ii) and (iii) of Claim~\ref{helper} continue to hold. They key difference is that
we obtain an improved expression for $p_{v,j}$ compared to the one in Lemma~\ref{probv}.
We now have that $p_{v,j}$ is almost $o'_{v,j}$ since instead of sampling $k$ trees
independently as in Algorithm~\ref{multialg} 
(which incurs a loss since $\prod_i(1-a_i)$ is smaller than $\sum_i a_i$), we now sample a 
{\em single tree configuration}; this improved bound also results in the improved
approximation.    
More precisely, mimicking the proof of Lemma~\ref{probv}, we obtain that 
$p_{v,j}\leq\frac{o'_{v,j}}{\kp}+\bigl(1-\frac{1}{\kp}\bigr)p_{v,j-1}$ for all $v$ and $j\geq -1$. 
Plugging this in the proof of Lemma~\ref{latv} gives
$\E[L_v]\leq\frac{c+1}{(\ln c)(1-c(1-1/\kp))}\cdot\lat_v\leq\frac{c+1}{(\ln c)(1-c\e)}\cdot\lat_v$ 
for all $v$. 

The expression $\frac{c+1}{\ln c}$ achieves its minimum value of $\mu^*$ at $c=\mu^*$
(i.e., when $c+1=c\ln c$), so the approximation factor is at most $\frac{\mu^*}{1-\mu^*\e}$.

\appendix

\section{Proofs from Section~\ref{prelim}} \label{append-prelim} \label{append-bns}

\begin{proofof}{Corollary~\ref{cgcor}}
We argue that the length $L$ of the shortest $1\leadsto n$ path in
$\cg\bigl(f(1),\ldots,f(n)\bigr)$ is at most the claimed bound, which implies the claimed
statement. 

For $x\in[1,1+k(n-1)]$, define $f_k(x):=f\bigl(1+\frac{x-1}{k}\bigr)$. 
Note that for any $x\in[1,1+k(n-1)]$, $(x,f_k(x))$ is a corner point of $f_k$ iff
$(x',f(x'))$, where $x'=1+\frac{x-1}{k}$ is a corner point of $f$.
Also, if $(x',f(x'))$ is a corner point of $f$, then $x'$ must be an integer and
$f(x')=C_{x'}$. Hence, if $(x,f_k(x))$ is a corner point of $f_k$, then $1+\frac{x-1}{k}$
must be an integer.

Now consider the shortest $1\leadsto N:=1+k(n-1)$ path $P_k$ in
$\cg\bigl(f_k(1),f_k(2),\ldots,f_k(N)\bigr)$. Let $L_k$ be the length of $P_k$.
Consider an edge $(o,\ell)$ of $P_k$. Let
$o'=1+\frac{o-1}{k}$, $\ell'=1+\frac{\ell-1}{k}$. By Theorem~\ref{cgthm} and the above
discussion, $o'$ and $\ell'$ must be integers. The cost of $(o,\ell)$ is
\begin{equation*}
\begin{split}
f_k(\ell)\Bigl(N-\frac{o+\ell}{2}\Bigr) & 
= f(\ell')\Bigl(k(n-1)-\frac{o-1+\ell-1}{2}\Bigr) \\
& = k\cdot f(\ell')\Bigl(n-1-\frac{(o-1)/k+(\ell-1)/k}{2}\Bigr)
= k\cdot f(\ell')\Bigl(n-\frac{o'+\ell'}{2}\Bigr)
\end{split}
\end{equation*}
which is the $k$ times the cost of the $(o',\ell')$ edge in
$\cg\bigl(f(1),\ldots,f(n)\bigr)$. 
Thus, $L\leq L_k/k$ for all $k\geq 1$.

Moreover, 
$$
\frac{L_k}{\mu^*/2}\leq\sum_{x=1}^Nf_k(x)\leq\sum_{x=2}^N\Bigl(\int_{x-1}^xf_k(t)dt+f_k(x)-f_k(x-1)\Bigr)
=\int_1^Nf_k(t)dt+f_k(N)=k\int_1^nf(x)dx+f(n).
$$
Therefore, $\frac{L}{\mu^*/2}\leq\int_1^nf(x)dx+\frac{f(n)}{k}$ for all $k\geq 1$, which
implies that $L\leq\frac{\mu^*}{2}\int_1^nf(x)dx$.
\end{proofof}

\begin{proofof}{Theorem~\ref{bnsbnd}}
Part (ii) for \mlp follows from the analysis in~\cite{ArcherB10}, so we focus on part
(i). We use a concatenation-graph argument similar to the one used for single-vehicle \mlp
in~\cite{GoemansK98}.  
Let $b^*_\ell=\optbns{k,\ell}$ for $\ell=1,\ldots,n$. 
Consider any sequence $0<\ell_1<\ell_2<\ldots<\ell_h=n$ of indices. 
We can obtain a solution from this as follows. For each index $\ell=\ell_p$ and each
$i=1,\ldots,k$, we double and shortcut the $r_i$-rooted path in the optimal solution to
the $(k,\ell)$-bottleneck stroll problem and traverse the resulting cycle in a random
direction to obtain a tour $Z_{i,\ell}$ (which contains $r_i$).
For every $i=1,\ldots,k$, we then concatenate the tours
$Z_{i,\ell_1},\ldots,Z_{i,\ell_h}$.
Since $\ell_h=n$, this covers all nodes, so we obtain a feasible solution to the
multi-depot \kmlp instance. 

We show that the cost of this solution is at most the length of the
$1\rightarrow\ell_1\rightarrow\ldots\rightarrow\ell_h$ path in the concatenation graph
$\cg(2b^*_1,\ldots,2b^*_n)$, so the statement follows from Theorem~\ref{cgthm}.

To bound the cost, consider an edge $(o,\ell)$ of the path. 
Suppose inductively that we have covered at least $o$ nodes by the partial solution
constructed by concatenating tours corresponding to the nodes on the path up to and
including $o$.
Consider the additional contribution to the total latency when we concatenate tour
$Z_{i,\ell}$ to vehicle $i$'s current route, for $i=1,\ldots,k$. 
The resulting partial solution covers at least $\ell$ nodes (since 
$\bigcup_i V(Z_{i,\ell})\geq\ell$). Suppose 
that in this step we cover $B$ additional nodes, and 
there are $A$ uncovered nodes remaining after this step.
Then, $B\leq\ell-o$ and $A\leq n-\ell$.
The latency of the uncovered increases by at most $\max_i c(Z_{i,\ell})\leq 2b^*_\ell$. 
The latency of each node $u$ that got covered by, say, $Z_{i,\ell}$ increases by at most
$\frac{c(Z_{i,\ell})}{2}\leq b^*_\ell$ since we choose a random direction for traversing
the cycle $Z_{i,\ell}$.
Therefore, the total increase in latency is at most 
$$
2b^*_{\ell}\Bigl(A+\frac{B}{2}\Bigr)\leq 2b^*_{\ell}\Bigl(n-\ell+\frac{\ell-o}{2}\Bigr)
=2b^*_{\ell}\Bigl(n-\frac{o+\ell}{2}\Bigr).
$$
This is precisely the length of the $(o,\ell)$ edge in $\cg(2b^*_1,\ldots,2b^*_n)$, and so
by induction, the total latency of the solution is at most the length of the
$1\rightarrow\ell_1\rightarrow\ldots\rightarrow\ell_h$ path in
$\cg(2b^*_1,\ldots,2b^*_n)$. 
\end{proofof}

\section{Proof of Theorem~\ref{arbpoly}} \label{append-arbpoly}

We restate the theorem for easy reference.

\newtheorem*{thm_arbpoly_restated}{Theorem~\ref{arbpoly}}
\begin{thm_arbpoly_restated}
Let $D=(U+r,A)$ be a digraph with nonnegative integer edge weights $\{w_e\}$, 
where $r\notin U$ is a root node, such that $|\delta^{\into}(u)| \ge |\delta^{\out}(u)|$
for all $u \in U$. 
For any integer $K\geq 0$, one can find out-arborescences $F_1,\ldots,F_q$
rooted at $r$ and integer weights $\gm_1,\ldots,\gm_q$ in polynomial time such that  
$\sum_{i=1}^q\gm_i=K$, $\sum_{i:e\in F_i}\gm_i\leq w_e$ for all $e\in A$, 
and $\sum_{i:u\in F_i}\gm_i \ge\min\{K,\ld_D(r,u)\}$ for all $u\in U$.
\end{thm_arbpoly_restated}

We require some notation and lemmas proved by~\cite{BangjensenFJ95}. 
To avoid confusion we use the superscript ${}^*$ when referring to statements in~\cite{BangjensenFJ95}.
Theorem~\ref{arbpoly} is a polytime version of Corollary$^*$ 2.1 in~\cite{BangjensenFJ95},
and our proof closely follows that of Theorem$^*$ 2.6 in Bang-Jensen et al.~\cite{BangjensenFJ95}.
Let $\Lambda_D(u,v) = \min\{K,\lambda_D(u,v)\}$ be the required connectivity.

\begin{definition} 
Let $e=(t,u)$ and $f=(u,v)$ be edges. \emph{Splitting off} $e$ and $f$ means removing $e$
and $f$ and adding a new edge $(t,v)$, or, in a weighted graph, subtracting some amount
$x>0$ from the weights $w_e$ and $w_f$ and increasing $w_{(t,v)}$ by $x$. We denote the new
digraph by $D^{ef}$. Edges $e$ and $f$ are \emph{splittable} if $\lambda_{D^{ef}}(x,y) \ge
\Lambda_D(x,y)$ for all $x,y \neq u$. 
\end{definition}

Say that $u$ and $v$ are {\em separated} by $X$ is $|X\cap\{u,v\}|=1=|\{u,v\}\sm X|$.
We call a set of nodes $X$ \emph{tight} if 
$\min\{|\delta^\into(X)|,|\delta^\out(X)|\} =\max_{u, v\text{ separated by }X} \Lambda_D(u,v)$,
that is, $X$ is a minimum cut for some $u, v$ maximum flow, 
and we say that $X$ is tight for $u, v$ if $u, v$ are separated by $X$ and
$\min\{|\delta^\into(X)|,|\delta^\out(X)|\} = \Lambda_D(u,v)$. 
If $t,v \in X$ and $u \notin X$, splitting edges $(t,u),(u,v)$ reduces $|\delta^\into(X)|$
and $|\delta^\out(X)|$, so there is a close relationship between splittable edges and
tight sets. 
Note that for an Eulerian digraph $D=(V,A)$ (i.e., $|\delta^\into(u)| = |\delta^\out(u)|$ for all 
$u\in V$), we have $|\dt^\into(X)|=|\dt^\out(X)|$ for all $X\sse V$, and
$\Ld_D(u,v)=\Ld_D(v,u)$ for all $u,v\in V$.

\begin{lemma}[Claim$^*$ 2.1 in~\cite{BangjensenFJ95}]
\label{lem_tightset}
Edges $e= (t,u)$ and $f= (u,v)$ are splittable if and only if there is no set $X$ such
that $t,v \in X$, $u \notin X$, and $X$ is tight for some $x,y \neq u$. 
\end{lemma}

Lemma \ref{lem_tightset} can be used to prove that splittable edges always exist.

\begin{lemma}[Theorem$^*$ 2.2 in~\cite{BangjensenFJ95}]
\label{lem_splittable}
Let $D = (V,A)$ be an Eulerian digraph 
and $v \in V$ with $|\delta^\out(v)| \neq 0$. Then for every edge $f=(u,v)$ there
is an edge $e=(t,u)$ such that $e$ and $f$ are splittable. 
\end{lemma}

Bang-Jensen et al.'s exponential version of Theorem \ref{arbpoly} repeatedly splits off unweighted pairs of edges and recurses on the new graph. We follow the same procedure but always split the same pair as many times/as much weight as possible at once.
The following simple observation allows us to prove this runs in polynomial time.

\begin{lemma}
\label{lem_tightset_monotonicity}
Let $D = (V,A)$ be an Eulerian digraph, $e=(t,u)$ and $f=(u,v)$ be splittable edges, $f'=(u,v')$ be an edge leaving $u$ (possibly $f=f'$), and $X_{f'}$ be a tight set for some $x,y\neq u$ with $u \notin X_{f'}$, $v' \in X_{f'}$. 
Then $X_{f'}$ is still tight for $x,y$ in $D^{ef}$ after splitting off $e$ and $f$.
\end{lemma}

\begin{proof}
We have that $\lambda_{D^{ef}}(x,y) \geq \Lambda_D(x,y)$, since $e$ and $f$ are splittable
and $x,y\neq u$. Splitting off cannot increase $\ld(x,y)$, so $\Ld_{D^{ef}}(x,y)=\Ld_D(x,y)$.
Splitting off does not affect $|\delta^\into(X_{f'})|$ or $|\delta^\out(X_{f'})|$
unless $t,v \in X_{f'}$, 
and by Lemma \ref{lem_tightset} this
cannot be the case since $X_{f'}$ is tight and $e,f$ are splittable. Therefore
$\Lambda_{D^{ef}}(x,y) = \min\{|\delta^\into(X_{f'})|,|\delta^\out(X_{f'})|\}$. 
\end{proof}

As a consequence, $O(n^2)$ splittings suffice to remove a node.

\begin{lemma}
\label{lem_poly_splittings}
Let $u$ be a node in an Eulerian digraph $D = (V,A)$, and suppose we repeatedly choose
$t,v$ such that $(t,u),(u,v)$ are splittable and split them off to the maximum extent
possible (i.e., maximum splittable weight). Then after $O(n^2)$ such splittings
$|\delta^\into(u)|$ and $|\delta^\out(u)|$ will be reduced to $0$. 
\end{lemma}

\begin{proof}
If $|\delta^\out(u)| > 0$, then by Lemma \ref{lem_splittable} there is a splittable pair $e=(t,u),f=(u,v)$. Splitting $e,f$ as much as possible creates a tight set $X_{f}$, and this set remains tight after additional splittings centered at node $u$ by Lemmas \ref{lem_tightset} and \ref{lem_tightset_monotonicity}, so $e,f$ will not become splittable again.
After $O(n)$ splittings $e',f$ for all possible $e'$, $w_f$ must be 0, and since there are
$O(n)$ choices for $f$ and $|\delta^\into(u)|=|\delta^\out(u)|$, $O(n^2)$ splittings
suffice to remove all edges incident to $u$. 
\end{proof}

\begin{proofof}{Theorem \ref{arbpoly}}
Note that $r \notin U$. If $|U| \le 1$ the theorem is trivial, so assume $|U| \ge 2$. For
every $u \in U$, add an edge $(u,r)$ of weight $|\delta^\into(u)|-|\delta^\out(u)|$, and let
$D'$ the resulting Eulerian graph (with $\lambda_{D'}(u,v) \ge \lambda_D(u,v)$).  

Let $u = \argmin_{v \in U} \Lambda(r,v)$, and $f=(u,v)$ be an edge leaving $u$ with $w_f > 0$.
By Lemma \ref{lem_splittable} there exists an edge $e=(t,u)$ such that $e$ and $f$ are splittable.
Let $x$ be the maximum amount $e$ and $f$ are splittable, which can be found by binary search. Split off $e$ and $f$ to an extent $x$ (subtract $x$ from $w_e,w_f$, add $x$ to $w_{(t,v)}$) and recurse on $D'^{ef}$.

Note that only $\Lambda_{D'^{ef}}(r,u)$ can decrease in the split, so in the recursive
call we will choose the same $u$ if $\Lambda_{D'^{ef}}(r,u) > 0$. By Lemma
\ref{lem_poly_splittings} $O(n^2)$ iterations suffice to remove all edges incident to
$u$. Future splittings after this centered on other nodes may create new edges but will
never add an edge incident to a node of degree $0$, so $O(n^3)$ splittings suffice. 

We now undo the splitting to construct the arborescence packing on the original $D$.
By induction we can find arborescences $F_1,\ldots, F_q$ in $D^{ef}$ and weights
$\gamma_1,\ldots,\gamma_q$ such that $\sum_{i=1}^q\gm_i=K$, $\sum_{i:h\in F_i}\gm_i\leq
w_{h,D^{ef}}$ for all $h\in A(D^{ef})$,  $\sum_{i:u'\in F_i}\gm_i \ge \Lambda_{D}(r,u')$
for all $u' \neq u$, and $\sum_{i:u\in F_i}\gm_i\ge \Lambda_{D^{ef}}(r,u) \ge
\Lambda_{D}(r,u) - x$.  
First, we need to ensure that the $F_i$ do not use the added arc $g=(t,v)$ above its weight in $D$, and second we need to update the arborescences to cover $u$ to an additional $x$ extent.

If $w_{g,D^{ef}} \ge \sum_{i:g\in F_i}\gm_i > w_{g,D} = w_{g,D^{ef}} - x$ we need to decrease the use of $g$ by $\sum_{i:g\in F_i}\gm_i - w_{g,D}$, which is at most $x$. We can replace $g$ with the pair $e,f$ since $w_{h,D} - x = w_{h,D^{ef}} \ge \sum_{i:h\in F_i}\gm_i$ for $h =e$ or $f$.
Repeatedly choose $F_i$ containing $g$ until we have a set $S$ with total weight at least $x$. Break the last $F_i$ added to $S$ into two identical arborescences with weights summing to $\gamma_i$, so that the set $S$ has weight exactly $x$. This increases $q$ by $1$ (or $O(n^3)$, summing over all graphs in the induction). 

For each $F_i \in S$, if $u \notin V(F_i)$, define $F_i' = F_i - g+e + f$. If $u \in V(F_i)$, let $h$ be the last edge on the path $P$ from $r$ to $u$ in $F_i$. If $g \notin P$, define $F_i' = F_i -g + f$, and if $a \in P$, define $F_i' = F_i - g -h +e + f$ (note $F_i'$ is connected). Replace each $F_i$ with $F_i'$ in the arborescence packing. Over all $F_i \in S$, we remove $g$ from trees with weight $x$ and add $e$ and $f$ to trees with weight at most $x$.
The updated $F'_1,\ldots, F'_{q+1}$ now satisfy $\sum_{i:h\in F_i}\gm_i\leq w_{h,D}$.

If $\Lambda_{D}(r,u) > \sum_{i:u\in F'_i}\gm_i$ we need to increase the weight of some $F'_i$ containing $u$. 
By assumption $\sum_{i:u\in F_i}\gm_i \ge \Lambda_{D^{ef}}(r,u) \ge \Lambda_{D}(r,u) - x$.
We chose $u$ such that $\Lambda_D(r,u) \le \lambda_D(r,v)$ for all $v$, so $\Lambda_{D^{ef}}(r,u) \le \lambda_{D^{ef}}(r,v) - x$. Therefore for every $v \neq u$ there are $F_i$ containing $v$ but not $u$ with total weight at least $x$ (or $x - (\sum_{i:u\in F'_i}\gm_i - \Lambda_{D^{ef}}(r,u))$ if $u$ is in more $F_i$ than required).

Let $S_2$ be a set of $F'_i$ with weight at least $x$ containing $t$ but not $u$. Add the edge $e$ to $F'_i \in S_2$. In the process we may exhaust the budget $w_e - \sum_{i:u\in F'_i}\gm_i$ for $e$ due to adding $e$ to some $F_i \in S$ in the previous step.
This can happen for two reasons. In the first case, there was some $F_i \in S$ with $u \notin V(F_i)$, and we defined $F_i' = F_i - g+e + f$. But that change also increased the weight of $F'_i$ containing $u$ beyond $\Lambda_{D^{ef}}(r,u)$ and is not a problem. 

The second case is that $u \in F_i$, and we set $F_i' = F_i - g -h +e + f$, which uses budget for $e$ even though $u$ is already in $F'_i$. However, at the same time we lost $y$ budget for $e$, we freed up $y$ budget for some $h=(s,u)$, and we can find a set of $F_i$ with weight $y$ containing $s$ but not $u$ and add the edge $h$.
This step may require breaking some $F_i$ into two trees with total weight $\gamma_i$ that are identical except that one contains $u$ and the other does not, which may increase $q$ by $O(n)$, but it remains polynomially bounded ($O(n^4)$ added over the entire induction).
\end{proofof}

\section{Proof of Corollary~\ref{pccor}} \label{append-pccor}

\begin{proofof}{part (i)}
We utilize part (ii) of Theorem~\ref{pcstroll}. We may assume that when $\ld=L=0$, the
tree $T_\ld$ returned is the trivial tree consisting only of $\{r\}$, and when $\ld$ is
very large, say $H=nc_{\max}$, then $T_\ld$ spans all nodes. So if $B=0$ or $n$,
then we are done, so assume otherwise. 
Let $\C^*_B$ be an optimal collection of rooted paths, so $\iopt=\sum_{P\in\C^*_B}c(P)$.
Let $n^*=|\bigcup_{P\in\C^*_B}V(P)|\geq B$.
We preform binary search in $[L,H]$ to find a value $\ld$ such that $|V(T_\ld)|=B$. 
We maintain the invariant that we have trees $T_1, T_2$ for the endpoints $\ld_1<\ld_2$
of our search interval respectively such that $|V(T_1)|<B<|V(T_2)|$ and 
$c(T_i)+\ld_i|V\sm V(T_i)|\leq\iopt+\ld_i(n-n^*)$ for $i=1,2$.
Let $\ld=(\ld_1+\ld_2)/2$, and let $T=T_\ld$ be the tree returned by Theorem~\ref{pcstroll}
(ii). If $|V(T)|=B$, then we are done and we return the rooted tree $T$. 
Otherwise, we update $\ld_2\assign\ld$ if $|V(T)|>B$, and update $\ld_1\assign\ld$
otherwise. 

We terminate the binary search when $\ld_2-\ld_1$ is suitably small.
To specify this precisely, consider the parametric LP \eqref{bns_lp} where $\pi_v=\ld$ for
all $v\in V$ and $\ld$ is a parameter. 
We say that $\ld$ is a {\em breakpoint} if there are two optimal solutions $(x^1,z^1)$,
$(x^2,z^2)$ to \eqref{bns_lp} with $\sum_vz^1_v\neq\sum_vz^2_v$. (This is equivalent to
saying that the slope of the optimal-value function is discontinuous at $\ld$.)
We may assume that $(x^1,z^1)$, $(x^2,z^2)$ are vertex solutions and so their non-zero
values are multiplies of $\frac{1}{M}$ for some $M$ (that can be estimated) with 
$\log M=\poly(\text{input size})$. But then $\sum_e c_ex_e$ and $\sum_vz_v$ are also
multiples of $\frac{1}{M}$ for both solutions (since the $c_e$s are integers), and hence
the breakpoint $\ld$ is a multiple of $\frac{1}{M'}$, for some $M'\leq nM$. We terminate the
binary search when $\ld_2-\ld_1<\frac{1}{2n^2M}$. Observe that the binary search takes
polynomial time. 

So if we do not find $\ld$ such that $|V(T_\ld)|=B$, at termination, we have that 
$c(T_i)+\ld_i(n-|V(T_i)|)\leq\iopt+\ld_i(n-n^*)$ for $i=1,2$. There must be exactly one
breakpoint $\ld\in[\ld_1,\ld_2]$. There must be at least one breakpoint since 
$T_1\neq T_2$, which can only happen if the optimal solutions to \eqref{bns_lp} differ for
$\ld_1$ and $\ld_2$, and there cannot be more than one breakpoint since any two
breakpoints must be separated by at least $\frac{1}{nM}$ as reasoned above. 

We claim that we have $c(T_i)+\ld(n-|V(T_1)|)\leq\iopt+\ld(n-n^*)$ for $i=1,2$. If we
show this, then taking $a, b$ so that $a|V(T_1)|+b|V(T_2)|=B$, $a+b=1$, and
taking the $(a,b)$-weighted combination of the two inequalities, we obtain that
$ac(T_1)+bc(T_2)+\ld(n-B)\leq\iopt+\ld(n-n^*)$, and so we are done. (Note that we do
{\em not} actually need to find the breakpoint $\ld$.)

To prove the claim observe that 
\begin{equation*}
\begin{split}
 c(T_1)+\ld(n-|V(T_1)|) & \leq c(T_1)+\ld_1(n-|V(T_1)|)+\frac{1}{2nM} \\
& \leq\iopt+\ld_1(n-n^*)+\frac{1}{2nM}
\leq\iopt+\ld(n-n^*)+\frac{1}{2nM}. 
\end{split}
\end{equation*}
So $\bigl[c(T_1)+\ld(n-|V(T_1)|)\bigr]-\bigl[\iopt+\ld(n-n^*)\bigr]\leq\frac{1}{2nM}$, but
the LHS is a multiple of $\frac{1}{M'}$, so the LHS must be nonpositive. A similar argument
shows that $c(T_2)+\ld(n-|V(T_2)|)\leq\iopt+\ld(n-n^*)$.
\end{proofof}

\begin{proofof}{part (ii)}
We mimic the proof of part (i), and only discuss the changes. Assume that
$0<C<$(cost of MST of $\{v: w_v>0\}$) to avoid trivialities.
Let $\W^*_C$ be an optimal collection of rooted paths, so
$n^*=w(\bigcup_{P\in\W^*_C}V(P))$. Let $\iopt=\sum_{P\in\W^*_C}c(P)\leq C$.
Let $K$ be such that all $w_v$s are multiples of $\frac{1}{K}$; note that  
$\log K=\poly(\text{input size})$. 
Let $W=\sum_v w_v$.
For a given parameter $\ld$, we now consider \eqref{bns_lp} with penalties $\ld w_v$ for 
all $v$. We perform binary search in $[L=0,H=KWc_{\max}]$; we may again assume
that $T_L$ is the trivial tree, and $T_H$ spans all nodes with positive weight.
Given the interval $[\ld_1,\ld_2]$, we maintain that the trees $T_1,T_2$ for $\ld_1,\ld_2$
satisfy $c(T_1)<C<c(T_2)$ and $c(T_i)+\ld_iw(V\sm V(T_i))\leq\iopt+\ld_i(W-n^*)$ for
$i=1,2$. As before, we find tree $T_\ld$ for $\ld=(\ld_1+\ld_2)/2$.
If $c(T_\ld)=C$ then $w(V(T_\ld))\geq n^*$ and we are done and return $T_\ld$.
Otherwise, we update $\ld_2\assign\ld$ if $c(T)>C$, and $\ld_1\assign\ld$.
We terminate when $\ld_2-\ld_1\leq\frac{1}{2W^2K^2M}$.

Similar to before, one can argue that every breakpoint of the parametric LP with penalties
$\{\ld w_v\}$ must be a multiple of $\frac{1}{M'}$ for some $M'\leq MKW$. So at
termination (without returning a tree), there is a breakpoint $\ld\in[\ld_1,\ld_2]$.
We have that for $i=1,2$,
$$
\Bigl[c(T_i)+\ld\bigl(W-w(V(T_i))\bigr)\Bigr]-
\Bigl[\iopt+\ld(W-n^*)\Bigr]\leq\frac{1}{2WK^2M}
$$
but is also a multiple of $\frac{1}{KM'}$, so the above quantity must be nonpositive for
$i=1,2$.
Let $a, b$ be such that $a+b=1$ and $ac(T_1)+bc(T_2)=C$. Then, we have
$$
a\Bigl[c(T_1)+\ld\Bigl(W-w\bigl(V(T_1)\bigr)\Bigr)\Bigr]+
b\Bigl[c(T_2)+\ld\Bigl(W-w\bigl(V(T_2)\bigr)\Bigr)\Bigr]
\leq\iopt+\ld(W-n^*).
$$
So $aT_1+bT_2$ yields the desired bipoint tree.
\end{proofof}

\section{Weighted sum of latencies: a combinatorial $7.183$-approximation for
  \kmlp} \label{combwtextn} 
We may assume via scaling that all weights are integers, and $w_r=1$. 
Let $W=\sum_{u\in V} w_u$. 
As before, let $\TS:=\{\Time_{-1}=0,\Time_0,\ldots,\Time_D\}$, where
$\Time_j=\ceil{(1+\e)^j}$ for all $j\geq 0$, and $D=O(\log\Time)=\poly(\text{input size})$
is the smallest integer such that $\Time_D\geq\Time$. 
We think of applying the concatenation-graph argument by considering time points 
in the polynomially-bounded set $\TS$. We obtain a suitable $k$-tuple of
tours for each of these time points, and concatenate the tours corresponding to the
shortest path in the corresponding concatenation graph.

As in Algorithm~\ref{kmlpalg}, we initialize $C\assign\{(1,0)\}$ and $\Qc\assign\es$.
For each time $t=0,1,\ldots$, let $\wt_t$ denote the maximum node weight of a collection
of $k$ rooted paths, each of length at most $t$. Note that $\wt_0=1$.
By part (ii) of Corollary~\ref{pccor}, given $t$, we can efficiently compute a rooted tree
$Q_t$ or rooted bipoint tree $(a_t,Q_t^1,b_t,Q_t^2)$ of cost $kt$ and node weight at least
$\wt_t$. 
We compute these objects 
for all $t\in\TS$. Say that $t$ is single or bipoint 
depending on whether we compute a tree or a bipoint tree 
for $t$ respectively.
Observe that given a rooted tree $Q$, one can obtain a collection of $k$ cycles, each of 
length at most $\frac{2c(Q)}{k}+2\max_{v\in V(Q)}c_{rv}$, that together cover $V(Q)$.

For each $t\in\TS$, if $t$ is single, we include the point
$\bigl(w(V(Q_t)),\frac{2c(Q_t)}{k}+2t\bigr)$ if $t$ is single and add $Q_t$ to $\Qc$; if
$t$ is bipoint, we add the points 
$\bigl(w(V(Q_t^1)),\frac{2c(Q_t^1)}{k}+2t\bigr)$,
$\bigl(w(V(Q_t^2)),\frac{2c(Q_t^2)}{k}+2t\bigr)$ to $C$ and add the trees $Q_t^1$, $Q_t^2$ 
to $\Qc$. 
Let $f:[1,W]\mapsto\R_+$ be the lower-envelope curve of $C$. 
Note that $f(1)=0$ and $f(W)\leq 4\Time_D$.
By Theorem~\ref{cgthm}, the shortest path $P_C$ in the concatenation graph
$\cg\bigl(f(0),f(1),\ldots,f(W)\bigr)$ only uses nodes corresponding to corner points of
$f$, all of which must be in $C$. So $P_C$ can be computed by finding the
shortest path in the polynomial-size (and polytime-computable) subgraph of the
concatenation graph consisting of edges incident to nodes corresponding to corner points
of $f$.

We now argue that we can obtain a solution of cost at most the length of $P_C$ in the 
concatenation graph. Suppose that $(o,\ell)$ is an edge of the shortest path. 
Since $\bigl(\ell,f(\ell)\bigr)$ is a corner point, we have a tree $Q_\ell\in\Qc$ and time 
$t_\ell\in\TS$ such that: 
$w(V(Q_\ell))=\ell$, $f(\ell)=\frac{2c(Q_\ell)}{k}+2t_\ell$, and 
$\max_{v\in V(Q_\ell)}c_{rv}\leq t_\ell$.
Thus, as noted earlier, we can obtain from $Q_\ell$ a collection of $k$ cycles, each of
length at most $f(\ell)$ that together cover nodes of total weight $\ell$. 
One can then mimic the inductive argument in Theorem~\ref{kmlpround} to prove the claimed
result. 

By Theorem~\ref{cgthm}, the length of $P_C$ is at most
$\frac{\mu^*}{2}\sum_{\ell=1}^Wf(\ell)$. Define $b^*_\ell=\min\{t:\wt_t\geq\ell\}$,
similar to the definition of $\optbns{k,\ell}$.
So $\sum_{\ell=1}^Wb^*_\ell$ is a lower bound on the optimal value.
For an integer $x\geq 0$, recall that $[x]$ denotes $\{1,\ldots,x\}$; 
let $\bb{x}$ denote $\{0\}\cup[x]$.
Define $\Time_{-1}:=0$. 
Dovetailing the proof of Lemma~\ref{cglpbnd}, we have 
\begin{equation}
\frac{\sum_{\ell=1}^W f(\ell)}{4} = \int_{t=0}^f(W)dt\bigl|\{\ell\in[W]: f(\ell)\geq 4t\}\bigr| 
\leq \sum_{j=0}^D(\Time_j-\Time_{j-1})\bigl|\{\ell\in[W]: f(\ell)>4\Time_{j-1}\}\bigr|.
\label{ineq2}
\end{equation}
Consider some $t\in\TS$.
Recalling how the point-set $C$ is produced, we observe that:
\begin{list}{(\alph{enumi})}{\usecounter{enumi} \topsep=0ex \itemsep=-0.2ex}
\item if $t$ is single, then $(w,4t)\in C$ for some $w\geq\wt_t$;
\item if $t$ is bipoint, then $(w,4t)$ lies in the convex hull of $C$ for
  some $w\geq\wt_t$.  
\end{list}
In both cases, this implies that $f(\wt_t)\leq 4t$.
So $\bigl|\{\ell\in[W]: f(\ell)>4\Time_{j-1}\}\bigr|\leq W-\wt_{\Time_{j-1}}$ for 
$j\geq 1$; this also holds when $j=0$ since $f(1)=0,\ \wt_0=1$.
Substituting this in \eqref{ineq2}, gives that $\frac{\sum_{\ell=1}^Wf(\ell)}{4}$ is at most
\begin{equation*}
\sum_{j=0}^D(\Time_j-\Time_{j-1})\bigl(W-\wt_{\Time_{j-1}}\bigr)
=\sum_{j=0}^D(\Time_j-\Time_{j-1})\sum_{\ell=\wt_{\Time_{j-1}}+1}^W 1
=\sum_{\ell=2}^W\sum_{j\in\bb{D}: \wt_{\Time_{j-1}}+1\leq\ell}(\Time_j-\Time_{j-1})
=\sum_{\ell=2}^W\Time_{B_\ell}
\end{equation*}
where $B_\ell=\min\{j\in\bb{D}:\wt_{\Time_j}\geq\ell\}$.
Clearly $\Time_{B_\ell}\leq(1+\e)b^*_\ell$, so  
$\frac{\sum_{\ell=1}^Wf(\ell)}{4}\leq(1+\e)\sum_{\ell=1}^Wb^*_\ell$.

\section{Node service times: a combinatorial $7.183$-approximation for
  \kmlp} \label{combndwtextn} 
Recall that we define the mixed length of a path or tree $Q$ to be $c(Q)+d(V(Q))$, and   
$c'$ is the directed metric $c'_{u,v}=c_{uv}+d_v$ for all $u, v$. So the $c'$-cost of an
out-tree rooted at $r$ is exactly its mixed length (since $d_r=0$). 
The changes to Algorithm~\ref{kmlpalg} are as follows.

\begin{list}{$\bullet$}{\itemsep=0ex \addtolength{\leftmargin}{-3ex}}
\item First, we sort nodes in increasing order of $c_{rv}+d_v$; let $u_1=r,u_2,\ldots,u_n$
be the nodes in this sorted order, and let $H_j=(U_j,E_j)$ be the subgraph induced by
$U_j:=\{u_1,\ldots,u_j\}$. 

\item In step S2, we use part (i) of Corollary~\ref{pccor} with the graph $H_j$ and the
mixed-length objective. 
If we get a rooted out-tree $Q_{j\ell}$, we add the point 
$\bigl(|V(Q_{j\ell})|,\frac{2c'(Q_{j\ell})}{k}+2c'_{r,u_j}\bigr)$ to $C$.
If we get a rooted bipoint out-tree, we 
add the points 
$\bigl(|V(Q^1_{j\ell})|,2\frac{c'(Q^1_{j\ell})}{k}+2c'_{r,u_j}\bigr)$ and
$\bigl(|V(Q^2_{j\ell})|,2\frac{c'(Q^2_{j\ell})}{k}+2c'_{r,u_j}\bigr)$ to $C$.
As before, we add the rooted tree or the constituents of the bipoint tree to $\Qc$. 

\item In step S5, we obtain a collection of $k$ tours $Z_{1,\ell},\ldots,Z_{k,\ell}$
from the rooted tree $Q^*_\ell$ by applying Lemma~\ref{break} on $Q^*_\ell$ with
$S=V(Q^*_\ell)$ to obtain $k$ cycles, and traverse each cycle in a random direction.
The resulting tours satisfy 
\begin{equation}
c(Z_{i,\ell})+2d(V(Z_{i,\ell}))\leq
2\cdot\frac{c(Q^*_\ell)+d\bigl(V(Q^*_\ell)\bigr)}{k}+
2\bigl(c_{ru_{j^*_\ell}}+d_{u_{j^*_\ell}}\bigr) \qquad \frall i=1,\ldots,k.
\end{equation}
\end{list}
The remaining steps of Algorithm~\ref{kmlpalg} (i.e., S1, S3, S4, S6) are unchanged.

The analysis follows the one in Section~\ref{comb}.
The bottleneck-$(k,\ell)$-stroll problem is now defined with respect to the mixed-length
objective; so $\optbns{k,\ell}$ is the smallest $L$ such that there are $k$ rooted paths
$P_1,\ldots,P_k$, each of mixed length at most $L$ that together cover at least $\ell$
nodes. As before, $\bnslb:=\sum_{\ell=1}^n\optbns{k,\ell}$ is a lower bound on the 
optimum. 

Lemma~\ref{sbound} holds as is. The argument for part (ii) is unchanged.
For part (i), suppose that $v_j$ is the node furthest from $r$ that is covered by some
optimal $(k,\ell)$-bottleneck-stroll solution, so $c_{ru_j}+d_{u_j}\leq\optbns{k,\ell}$. 
Then, in iteration $(j,\ell)$ of step S2, 
we add one or two points to $C$ such that the point 
$\bigl(\ell,\frac{2z}{k}+2(c_{ru_{j}}+d_{u_j})\bigr)$,
for some $z\leq k\cdot\optbns{k,\ell}$, lies in
the convex hull of the points added. Therefore, $s_\ell=f(\ell)\leq 4\cdot\optbns{k,\ell}$ 
since $f$ is the lower-envelope curve of $C$.
Finally, the proof that the cost of the solution returned is at most the length of the
shortest path in $\cg(s_1,\ldots,s_n)$ is essntially identical to the proof in the
LP-rounding $7.183$-approximation algorithm for \kmlp in Section~\ref{ndsrv}.

\end{document}